\documentclass[12pt]{article}
\usepackage{epsfig,amsfonts,amsthm}
\usepackage{amsmath,amssymb}
\usepackage{cite}
\usepackage[dvipsnames]{color}
\newcommand{\be}{\begin{equation}}
\newcommand{\ee}{\end{equation}}
\newcommand{\bea}{\begin{eqnarray}}
\newcommand{\eea}{\end{eqnarray}}

\newcommand{\doublet}[2]{ \left( \begin{array}{c}#1 \\ #2 \end{array}\right) }

\newcommand{\Z}{\mathbb{Z}}
\newcommand{\x}{{\mathsf{x}}}
\newcommand{\y}{{\mathsf{y}}}

\newtheorem{theorem}{Theorem}
\newtheorem{proposition}[theorem]{Proposition}
\newtheorem{lemma}[theorem]{Lemma}
\newtheorem{conjecture}[theorem]{Conjecture}

\def\lsim{\mathrel{\rlap{\lower4pt\hbox{\hskip1pt$\sim$}}
    \raise1pt\hbox{$<$}}}         
\def\gsim{\mathrel{\rlap{\lower4pt\hbox{\hskip1pt$\sim$}}
    \raise1pt\hbox{$>$}}}         

\title{Abelian symmetries in multi-Higgs-doublet models}

\author{Igor~P.~Ivanov$^{1,2}$, Venus~Keus$^{1}$, Evgeny Vdovin$^{2}$
\\
  {\small $^1$ IFPA, Universit\'{e} de Li\`{e}ge, All\'{e}e du 6 Ao\^{u}t 17, b\^{a}timent B5a, 4000 Li\`{e}ge, Belgium}\\
  {\small $^2$ Sobolev Institute of Mathematics, Koptyug avenue 4, 630090, Novosibirsk, Russia}\\
  }

\begin{document}
\maketitle

\begin{abstract}
$N$-Higgs doublet models (NHDM) are a popular framework to construct electroweak symmetry breaking mechanisms beyond the Standard model.
Usually, one builds an NHDM scalar sector which is invariant under a certain symmetry group.
Although several such groups have been used,
no general analysis of symmetries possible in the NHDM scalar sector exists.
Here, we make the first step towards this goal by classifying the elementary building blocks, namely the abelian symmetry groups,
with a special emphasis on finite groups.
We describe a strategy that identifies all abelian groups which are realizable as symmetry groups of the NHDM Higgs potential.
We consider both the groups of Higgs-family transformations only and the groups which also contain
generalized $CP$ transformations.
We illustrate this strategy with the examples of 3HDM and 4HDM and prove several statements for arbitrary $N$.
\end{abstract}

\section{Introduction}

The Higgs mechanism of the electroweak symmetry breaking (EWSB) has not yet been discovered experimentally,
although the LHC is expected to provide soon first definite clues about
the nature of the EWSB.
Many different variants of the Higgs mechanism have been proposed so far \cite{CPNSh}.
Moreover, even if one focuses on a specific Higgs mechanism beyond the Standard Model (SM),
one has many free parameters describing the scalar and Yukawa sectors,
all of them having phenomenological consequences.
Ideally, one would want to have a formalism that reconstructs the phenomenology of a given model
in its most general form, that is, in the entire space of the allowed free parameters.
Unfortunately, such a construction seems impossible for most non-minimal Higgs mechanisms.
In this situation the standard approach is to narrow down the freedom by imposing certain symmetries both on the Higgs sector
and on the Yukawa interactions. The problem is that in most cases it is not known a priori which symmetries can be used.
Many specific symmetry groups have been analyzed but their complete lists for every class of non-minimal Higgs model
are still lacking.

The general task of classifying symmetries in models beyond SM can be, roughly,
separated into two subtasks. First, one needs to know which symmetries can be implemented
in the scalar sector of the model and how these symmetries are broken under EWSB.
Second, one has to understand how these symmetry patterns can be encoded in the Yukawa sector
of the model and what are the resulting properties of the fermions.

In this paper we exclusively focus on the first task in multi-Higgs-doublet models,
which is an ambitious program on its own.
It is true that the symmetries to be obtained within the scalar sector can be broken by a sufficiently complicated
Yukawa sector. However, it does not diminish the importance of rigorous study within the scalar sector due to following reasons.
(1) At the tree level, the vacuum properties in the electroweak theory come entirely from the scalar potential, \cite{Haber}.
Fermions can influence the vacuum properties only via loop corrections. Therefore, it is logical
first to derive the tree-level features of the model, and next to see the effect of corrections.
(2) The Yukawa couplings do not necessarily have to break every symmetry found in the scalar sector.
It is easy to construct a Yukawa sector, for example by coupling all fermions to only one Higgs doublet,
which inherits many of the symmetries of the scalar sector. This construction can lead to
interesting phenomenological consequences, for example to the possibility of $\Z_p$-stabilized scalar dark matter,
\cite{ZpDM}. Which symmetry groups can be extended to the entire lagrangian is a separate question to study.
(3) Historically, one of the most studied model beyond SM is the two-Higgs-doublet model (2HDM), \cite{Lee,review2011}.
After decades of theoretical and phenomenological studies which included various forms of the Yukawa sector,
it became clear that there were several pressing questions (including the origin of $CP$-violation in 2HDM)
which required an in-depth look into the properties of the general scalar potential.
Only when sufficient progress in the scalar sector was made, did it become clear how to construct Yukawa sectors
incorporating additional $CP$-violating features, see examples in \cite{CPexamples}.

\subsection{Realizable symmetries}

When searching for the symmetry group which can be implemented in the scalar sector of a non-minimal Higgs model,
one must distinguish between {\em realizable} and {\em non-realizable} groups.
If it is possible to write a potential which is symmetric under a group $G$
but not symmetric under any larger symmetry group containing it, we call $G$ a realizable group
(the exact definition will be given in Section~\ref{realizable}).
A non-realizable group, on the contrary, automatically leads to a larger symmetry group of the potential.
This means that if we write a potential symmetric under a non-realizable $G$, we'll discover that
it is in fact symmetric under a larger symmetry group $H \supset G$, and there is no way to avoid that.
Roughly speaking, it means that the Higgs potential of a given model cannot accommodate overly complicated symmetry patterns.

Thus, the true symmetry properties of the potentials are reflected in realizable groups.

Just to give an example, cyclic groups such as $\Z_7$ and $\Z_{13}$
have been used in the context of the flavor symmetry problem, see e.g. \cite{frobenius} and references therein.
These symmetry groups are usually implemented in the space of fermion families, but sometimes they are
transferred to the Higgs sector of the model.
Anticipating the results of this paper, we can state that these symmetry groups are not realizable
in the three-Higgs-doublet model potential:
trying to impose such a symmetry, one ends up with a potential that has a continuous symmetry.

Therefore, if one focuses on the scalar sector of the model and aims at classifying the possible groups
which can be implemented in the scalar potential, it is natural to restrict one's choice to the realizable symmetry groups.

\subsection{Multi-Higgs-doublet models}

One particular non-minimal Higgs sector where the classification of possible symmetries of the scalar sector
has been established is the two-Higgs-doublet (2HDM), \cite{Lee,Haber,review2011}.
Although the general Higgs potential in 2HDM cannot be minimized explicitly with the straightforward algebra,
it can still be analyzed fairly completely with several basis-invariant tools that were developed in the last years.
The basic idea is that reparametrization transformations, or the Higgs-basis changes, link differently looking
but physically equivalent potentials to each other. This idea can be implemented
via the tensorial formalism at the level of Higgs fields, \cite{Botella,Davidson,Gunion,Oneil},
or via geometric constructions in the space of gauge-invariant bilinears, \cite{Sartori,Nagel,Maniatis,Nishi0,Ivanov0,generalizedCP}.
In particular, all possible groups of Higgs-family symmetries and generalized $CP$-transformations of the 2HDM potential have been found.
In addition, it was noted that if one disregards the $U(1)$ gauge transformations
the structure of the Higgs potential of 2HDM is invariant under a wider class of transformations, \cite{pilaftsis}.

It is natural to attempt to extend these successful techniques to $N$ Higgs doublets.
Some properties of the general NHDM potential were analyzed in \cite{Maniatis,Erdem,Barroso,Nishi1,Ferreira}.
In particular, the bilinear-space geometric approach was recently adapted to the NHDM in \cite{Ivanov1,Ivanov2}.
In what concerns the symmetries of the scalar sector of NHDM, several groups
have attacked various issues of this problem, \cite{Ferreira,Grimus,Olaussen,Ivanov2,Machado,Fukuyama,Adler}.
In particular, in \cite{Ferreira} an accurate and lengthy analysis of possible textures of the quartic interaction in the 3HDM
led to the list of one-parametric symmetries of the potential.
Although the method of \cite{Ferreira} could be extended to $N>3$, its practical implementation
relies on the set of all possible textures for the quartic coupling constants, which is not known {\em a priori}.

In \cite{Ivanov2} an attempt was made to describe possible symmetry groups of NHDM in the adjoint representation.
The symmetry properties of the vectors $M_i$ and $\Lambda_{0i}$ and the symmetric tensor $\Lambda_{ij}$ in the
expression (\ref{Vadjoint}) written below could be easily identified,
but the main difficulty was to take into account the non-trivial shape of the NHDM orbit space,
see details in Section \ref{subsection-potential}.
Mathematically, this led to the problem of classifying intersections of two orthogonal subgroups of $SO(N^2-1)$,
which was not solved in \cite{Ivanov2}.

It is therefore fair to state that despite several partial successes
the task of classifying all possible Higgs-family and generalized-$CP$ symmetries in NHDM is far from being accomplished.

In this paper we make a step towards solving this problem. We focus on abelian symmetries of the NHDM scalar potential
and describe a strategy that algorithmically identifies realizable abelian symmetry groups for any given $N$.
The strategy first addresses groups of unitary transformations in the space of Higgs families
and then it is extended to groups that include antiunitary (generalized $CP$) transformations.
The strategy also yields explicit examples of the potentials symmetric under a given realizable group.

The structure of the paper is the following.
In Section \ref{section-symmetries} we describe the Higgs potential of the general NHDM, discuss the group of
reparametrization transformations and introduce the notion of realizable symmetries.
In Section \ref{section-strategy} we describe the strategy that identifies all realizable abelian symmetry groups.
Then in Section \ref{section-3HDM-4HDM} we illustrate the general approach with the examples of 3HDM
and 4HDM.
Section \ref{section-NHDM} contains some further general results valid for all $N$.
Finally, in Section \ref{section-discussion} we make several remarks and draw our conclusions.
In the Appendices we prove the theorem about maximal abelian groups of $PSU(N)$
and work out in detail the groups containing antiunitary transformations in the case of 3HDM.

\section{Scalar sector of the $N$-Higgs doublet model}\label{section-symmetries}

\subsection{NHDM potential}\label{subsection-potential}

In the NHDM we introduce $N$ complex Higgs doublets with electroweak isospin $Y=1/2$:
\be
\phi_a = \doublet{\phi_a^+}{\phi_a^0}\,,\quad a=1,\dots , N\,.
\ee
Hermitian conjugation is denoted by dagger: $\phi_a^\dagger$.
The generic renormalizable Higgs potential in NHDM can be written in the tensorial form \cite{CPNSh,Haber}:
\be
\label{V:tensorial}
V = Y_{ab}(\phi^\dagger_a \phi_b) + Z_{abcd}(\phi^\dagger_a \phi_b)(\phi^\dagger_c \phi_d)\,,
\ee
where all indices run from 1 to $N$.
Coefficients in the quadratic and quartic parts of the potential are grouped into components of tensors $Y_{ab}$
and $Z_{abcd}$, respectively;
there are $N^2$ independent components in $Y$ and $N^2(N^2+1)/2$ independent components in $Z$.

Once the Higgs potential is given, the first task is to find its global minimum.
To this end, we replace the Higgs field operators by their vacuum expectation values,
and interprete (\ref{V:tensorial}) as a scalar function defined in the $2N$-complex-dimensional vector space $\mathbb{C}^{2N}$.
The electroweak gauge group is reduced to the global $SU(2)\times U(1)$ transformation group acting in this space.
Pick up a point $x \in \mathbb{C}^{2N}$ and apply a transformation $g \in SU(2)\times U(1)$.
All points $x^g = g\cdot x$ which can be reached in this way form a {\em gauge orbit}.
Since the potential is electroweak-symmetric by construction, its values are equal at all points throughout a chosen orbit.
Therefore, the Higgs potential can be throught of as a scalar function defined in this orbit space.

This leads to the question of how this orbit space can be described.
In principle, one can uniquely define an orbit by providing $N^2$ gauge-invariant bilinears $(\phi^\dagger_a \phi_b)$,
not all of them being algebraically independent.
An even more convenient way is to introduce the real vector
$r^\mu = (r_0, r_i)$ with
\be
r_0 = \sqrt{{N-1\over 2N}} \, \phi^\dagger_a \phi_a\,,\quad
r_i = {1\over 2}\phi^\dagger_a (\lambda_i)_{ab}\phi_b\,,
\ee
where $\lambda_i$ are the standard $SU(N)$ generators.
Then, the orbit space can be described completely by a list of $N$ algebraic equalities and inequalities
written for $r_0$ and $r_i$ with coefficients being the invariant $SU(N)$ tensors of rank $\le N$, \cite{Ivanov1}.
The Higgs potential can be rewritten as
\be
V = - M_0 r_0 - M_i r_i + {1 \over 2} \Lambda_{00}r_0^2 + \Lambda_{0i}r_0r_i + {1 \over 2}\Lambda_{ij}r_ir_j\,,\label{Vadjoint}
\ee
which, if needed, can be represented even more compactly as $V = - M_{\mu} r^\mu + \Lambda_{\mu\nu}r^\mu r^\nu/2$.
All the coefficients $Y_{ab}$ and $Z_{abcd}$ become now components of the vector $M_{\mu}$
and symmetric tensor $\Lambda_{\mu\nu}$, which transform under adjoint representation of the
group $GL(N,\mathbb{C})$.
For more details about the properties of the Higgs potential in the space of bilinears, see \cite{Ivanov2}.

\subsection{The group of reparametrization transformations}

When discussing symmetries of the potential, we focus on the reparametrization transformations,
which are non-degenerate linear transformations mixing different doublets $\phi_a$ but keeping invariant the kinetic term
(which includes interaction of the Higgs fields with the gauge sector of the model).
Alternatively, they can be defined as norm-preserving transformations of doublets that do not change the intradoublet structure.
In this work we do not use the more general transformations in the spirit of \cite{pilaftsis} but
focus on the ``classically defined'' reparametrization transformations.

A reparametrization transformation must be unitary (a Higgs-family transformation) or antiunitary (a generalized $CP$-transformation):
\be
U: \quad \phi_a \mapsto U_{ab}\phi_b\qquad \mbox{or} \qquad U_{CP} = U \cdot J:\quad \phi_a \mapsto U_{ab}\phi^\dagger_b\,,
\ee
with a unitary matrix $U_{ab}$. The transformation $J \equiv CP$ acts on doublets by complex conjugation and satisfies $J^2 = 1$.

Let us focus first on the unitary transformations $U$. A priori, such transformations form the group $U(N)$. However,
the overall phase factor multiplication is already taken into account by the $U(1)_Y$ from the gauge group.
This leaves us with the $SU(N)$ group of reparametrization transformations. Then, this group has a non-trivial center $Z(SU(N))= \Z_N$ generated
by the diagonal matrix $\exp(2\pi i/N)\cdot 1_N$, where $1_N$ is the identity matrix. Therefore, the group of {\em physically distinct}
unitary reparametrization transformations is the projective special unitary group
\be
G_u = PSU(N) \simeq SU(N)/\Z_N\,.\label{Gu}
\ee

Consider now antiunitary transformations $U_{CP} = U\cdot J$, $U\in U(N)$.
One can define an action of $J$ on the group $U(N)$ by
\be
J\cdot U \cdot J = U^*\,,\label{JactiononUN}
\ee
where asterisk denoted complex conjugation.
This action leaves invariant both the overall phase subgroup $U(1)_Y$ and the center of the $SU(N)$.
Therefore, we again can consider only $U_{CP} = U\cdot J$ such that $U\in PSU(N)$.
Once this action is defined, we can represent all distinct reparametrization transformations as a semidirect product
\be
G = PSU(N) \rtimes \Z_2\,.\label{G}
\ee
Recall that $G$ is a {\em semidirect product} of subgroups $A$ and $B$ with normal subgroup $A$ 
(this fact is usually denoted as $G=A\rtimes B$) if $A$ is normal in $G$ and the intersection $A\cap B$ is trivial (equals $\{1\}$).

\subsection{Realizable symmetry groups}\label{realizable}

A (reparametrization) {\em symmetry} of the Higgs lagrangian is such a reparametrization transformation
that leaves the potential invariant: $V(U_{(CP)}\phi) = V(\phi)$.
For any given potential all such transformations form a group, the automorphism group of the potential, which we denote as
$G_V = Aut(V)\cap G$ and which is obviously a subgroup of $G$ in (\ref{G}).

Let us stress that when we say a group $G_V$ is a symmetry group of the potential $V$, we mean that it is equal to, not just contained in,
$Aut(V)\cap G$. In \cite{Ivanov2} it was suggested to call such groups {\em realizable} symmetry groups.
In simple words, a symmetry group $G_V$ is realizable if one can write a potential symmetric under $G_V$ but not symmetric
under any overgroup  of $G_V$, i.e., over $H$ with $G_V\lneq H \leq G$.
For example, a potential which depends on the first doublet only via $(\phi_1^\dagger \phi_1)$
is obviously symmetric under the cyclic group $\Z_n$ of discrete phase rotations of this doublet for any $n$. However, these $\Z_n$'s
have no interest on their own because they trivially arise as subgroups of the larger symmetry group of this potential $U(1)$
describing arbitrary phase rotations.
It is this $U(1)$, not its individual subgroups $\Z_n$, which has a chance to be the realizable symmetry group of this potential.

Since in this paper we deal with groups $A_u$ containing only unitary transformations of Higgs families and with groups $A$ that can contain
antiunitary transformations, let us adopt the following convention:
a group $A_u \in G_u$ is called realizable if $A_u =  Aut(V)\cap G_u$,
and a group $A \in G$ is called realizable if $A =  Aut(V)\cap G$.

Note that in the case of 2HDM, where all possible reparametrization symmetries are known,
$\Z_2,\, (\Z_2)^2,\, (\Z_2)^3,\, U(1),\, U(1)\times \Z_2,\, SU(2)$, each of these symmetry groups is realizable,
which is most easily proved in the space of bilinears.

\section{Finding abelian groups}\label{section-strategy}

We now focus on abelian groups for NHDM and
describe in this Section the strategy that identifies all realizable abelian groups of the NHDM potential for any $N$.
This strategy can be outlined as follows: we first describe maximal abelian subgroups of $PSU(N)$,
then we explore their realizable subgroups, and finally we check
which of these groups can be extended by including antiunitary transformations.

\subsection{Heuristic explanations}

Before presenting the rigorous strategy which characterizes realizable abelian groups in NHDM,
we first give heuristic arguments which should facilitate understanding of the main idea.

Suppose we are given a potential $V$ of the $N$-doublet model and we want to find which phase rotations 
leave this potential invariant. Clearly, these phase rotations form a group which is a subgroup
of the group of all possible phase rotations of doublets, that is,
the group of all diagonal unitary $N\times N$ matrices acting in the space of doublets.
This group is $[U(1)]^N \subset U(N)$ and can be parametrized by $N$ parameters $\alpha_j \in [0,2\pi)$:
\be
\mbox{diag}[e^{i\alpha_1},\, e^{i\alpha_2},\, \dots ,\, e^{i\alpha_N}]\,.
\label{maximaltorusUN}
\ee
The potential is a collection of $k$ monomial terms each of the form $(\phi^\dagger_a \phi_b)$
or $(\phi^\dagger_a \phi_b)(\phi^\dagger_c \phi_d)$.
Upon a generic phase rotation (\ref{maximaltorusUN}), each monomial term gains its own phase rotation.
For example, $(\phi^\dagger_a \phi_b)$ with $a\not = b$ gains the phase $\alpha_b - \alpha_a$,
$(\phi^\dagger_a \phi_b)(\phi^\dagger_a \phi_c)$ with $a, b, c$ all distinct gains the phase $\alpha_b +\alpha_c - 2\alpha_a$, etc.
In short, each monomial gets a phase rotation which is a linear function of $\alpha$'s with integer coefficients:
$\sum_{j=1}^N m_j \alpha_j$. The vector of coefficients $m_j$ can be brought by permutation and overall sign change to one of the following forms:
\be
(1,\,-1,\,0,\,\dots)\,, \quad
(2,\,-2,\,0,\,\dots)\,, \quad
(2,\,-1,\,-1,\,\dots)\,, \quad
(1,\,1,\,-1,\,-1,\,\dots)\,,\label{fourtypes}
\ee
or a zero vector. Note that in all cases $\sum_j m_j = 0$.
Thus, the phase transformation properties of a given monomial are fully described by its vector $m_j$.
The phase transformation properties of the potential $V$, which is a colleciton of $k$ monomials, is characterized
by $k$ vectors $m_{1,j},\, m_{2,j},\, \dots,\, m_{k,j}$, each $m_{i,j}$ being of one of the types in (\ref{fourtypes}).

If we want a monomial to be invariant under a given transformation defined by phases $\{\alpha_j\}$,
we require that $\sum_{j=1}^N m_j \alpha_j = 2 \pi n$ with some integer $n$. If we want the entire potential to be invariant under a given 
phase transformation, we require this for each individual monomial. In other words, we require that 
there exist $k$ integers $\{n_i\}$ such that the phases $\{\alpha_j\}$ satisfy the following system of linear equations:
\be
\sum_{j=1}^N m_{i,j}\alpha_j = 2 \pi n_i\,,\quad \mbox{for all $1 \le i \le k$}\,.\label{systemUN}
\ee
Solving this system for $\{\alpha_j\}$ yields the phase rotations that leave the given potential $V$ invariant.

One class of solutions can be easily identified: if all $\alpha_j$ are equal, $\alpha_j = \alpha$, then (\ref{systemUN}) 
with $n_i = 0$ is satisfied for any $\alpha$. These solutions form the $U(1)$ subgroup inside $[U(1)]^N$ 
and simply reflect the fact that the potential is constructed from bilinears $(\phi^\dagger_a \phi_b)$.
These solutions become trivial when we pass from the $U(N)$ to the $SU(N)$ group of transformations.
However, there can exist additional solutions of (\ref{systemUN}). They form a group which remains non-trivial 
once we pass from $U(N)$ to $SU(N)$ and further to $PSU(N)$, which is the group of physically distinct 
Higgs-family reparametrization transformations.
It is these solutions which we are interested in.

In order to find these solutions, we note that a matrix with integer entries can be ``diagonalized'' 
by a sequence of elementary operations on its rows or columns: permutation, sign change, and addition
of one row or column to another row (column).
``Diagonalization'' for a non-square matrix means that the only entries $m_{i,j}$ that can 
remain non-zero are at $i = j$. After that, the system splits into $k$ equations on $N$ phases of the generic form
\be
m_{i,i}\tilde \alpha_i = 2\pi \tilde n_i\,, \quad \tilde \alpha_i \in [0,2\pi)\,,\quad \tilde n_i \in \Z\,,
\ee
with non-negative integer $m_{i,i}$.
If $m_{i,i}=0$, this equation has solution for any $\alpha_i$; thus $i$-th equation gives a factor $\Z$ to the 
symmetry group of the potential.
If $m_{i,i} =1$, then this equation has no non-trivial solution, and the $i$-th equation does not contribute to the symmetry group.
If $m_{i,i} =d_i > 1$, then this equation has $d_i$ solutions which are multiples of $\alpha_i = 2\pi/d_i$, 
and the $i$-th equation contributes the factor $\Z_{d_i}$ to the symmetry group.
The full symmetry group of phase rotations is then constructed as direct product of these factors.

Thus, the task reduces to studying which diagonal values of the matrix $m_{i,j}$ can arise in a model with $N$ doublets. 
For small values of $N$ this task can be solved explicitly,
while for general $N$ one must rely upon subtle properties of $m_{i,j}$ which stem from (\ref{fourtypes}).

\subsection{Maximal abelian subgroups of $PSU(N)$}

In the previous subsection we outlined the main idea of the classification strategy.
However we worked there in the group $U(N)$, while the reparametrization group is $PSU(N)$, which
essentially complicated the analysis.
Repeating this strategy for abelian subgroups of $PSU(N)$ is the subject of the remaining part of the text.

We start by reminding the reader of the definition of a maximal abelian subgroup.
A maximal abelian subgroup of $G_u$ (\ref{Gu}) is an abelian group that is not contained
in a larger abelian subgroup of $G_u$. A priori, there can be several maximal abelian subgroups in a given group.
Any subgroup of $G_u$ must be either a maximal one, or lie inside a maximal one.
Therefore, we first need to identify all maximal abelian subgroups of $G_u$ and then study their realizable subgroups.

If $G_u$ were $SU(N)$, then the situation would be simple. As we describe in Appendix A, all maximal abelian subgroups
of $SU(N)$ are the so called {\em maximal tori},
\be
[U(1)]^{N-1} = U(1)\times U(1) \times \cdots \times U(1)\,.
\label{maximaltorus1}
\ee
All maximal tori are conjugate inside $SU(N)$,
that is, given two maximal tori $T_1$ and $T_2$
there exists $g \in SU(3)$ such that $g^{-1}T_1g=T_2$.
It means that without loss of generality one could pick up a specific maximal torus, for example,
the one that is represented by phase rotations of individual doublets
\be
\mbox{diag}[e^{i\alpha_1},\, e^{i\alpha_2},\, \dots ,\, e^{i\alpha_{N-1}},\, e^{-i\sum\alpha_i}]\,,
\label{maximaltorus2}
\ee
and study its subgroups. The analysis would then proceed essentially as we explained in the previous subsection with an additional condition
that all $\alpha$'s sum to zero.

However, the group of distinct reparametrizations is $G_u=PSU(N)$, and it has a richer structure.
Consider the canonical epimorphism (i.e., surjective homomorphism)
\be
\overline{\phantom{G}}:\ SU(N)\to SU(N)/Z(SU(N))=\overline{SU(N)} \simeq PSU(N)\,,\label{epi}
\ee
where $Z(SU(N)) \simeq \Z_N$ is the center of $SU(N)$.
Given subgroup $\overline{A}$ of $PSU(N)$ we denote by $A \in SU(N)$
the complete preimage of $\overline{A}$, i.e., $A=\{g\in SU(N)\mid \overline{g}\in \overline{A}\}$.
If $\overline{A}$ is abelian, its complete preimage $A$ {\em does not have to be abelian}.
Specifically, in Appendix~\ref{appendix-maximal-abelian-PSU} we prove the following Theorem:
\begin{theorem}\label{AbelianSubgroupsPSUN}
Let $\overline{A}$ be a maximal abelian subgroup of $PSU(N)$, denote by $A$ the
complete preimage of $\overline{A}$ in $SU(N)$. Then one of the following
holds:
\begin{itemize}
 \item[{\em (1)}] $A$ is abelian and $A$ is conjugate to the subgroup of all
diagonal matrices in $SU(N)$.
\item[{\em (2)}] $A$ is a finite nilpotent group of class $2$ (i.e., the commutator $[A,A]$ lies in the center $Z(A)$), the exponent of
$A$ is divisible by $N$ and divides~$N^2$. Moreover, $Z(A)=Z(SU(N))$.
\end{itemize}
\end{theorem}
This theorem states that there are two sorts of maximal abelian groups in $PSU(N)$:
(1) maximal tori, which will be constructed below, (2)
certain finite abelian groups, which are not subgroups of maximal tori
and which must be treated separately.

\subsection{Maximal tori in $PSU(N)$}
Let us explicitly construct the maximal torus in $PSU(N)$ referred to in Theorem 1.
We first introduce some convenient notation. A diagonal unitary matrix acting in the space of Higgs doublets
and performing phase rotations of individual doublets, such as (\ref{maximaltorus2}),
will be written as a vector of phases:
\be
\left(\alpha_1,\,\alpha_2,\,\dots,\,\alpha_{N-1}, -\sum\alpha_i\right)\,.
\ee
In addition, if $M_1,\ldots,M_k$ are subsets of a group $G$, then $\langle M_1,\ldots,M_k\rangle$ denotes the subgroup generated by $M_1\cup\ldots\cup
M_k$ of $G$.
Then we construct a maximal torus in $SU(N)$
\be
T_0 = \langle U(1)_1,U(1)_2 , \cdots , U(1)_{N-1}\rangle\,,\label{T0}
\ee
where
\bea
U(1)_1 & = & \alpha_1(-1,\, 1,\, 0,\, 0,\, \dots,\, 0)\,,\nonumber\\
U(1)_2 & = & \alpha_2(-2,\, 1,\, 1,\, 0,\, \dots,\, 0)\,,\nonumber\\
U(1)_3 & = & \alpha_3(-3,\, 1,\, 1,\, 1,\, \dots,\, 0)\,,\nonumber\\
\vdots &  & \vdots\nonumber\\
U(1)_{N-1} & = & \alpha_{N-1}(-N+1, \, 1,\, 1,\, 1,\, \dots,\, 1)\,,
\label{groupsUi}
\eea
with all $\alpha_i \in [0,2\pi)$.
Clearly, for every $j$ we have $\langle U(1)_1,\ldots,U(1)_{j-1}\rangle\cap
U(1)_j=\{e\}$, whence $T_0$ is a direct product of $U(1)_1,\ldots,U(1)_{N-1}$.
In particular, any element $u \in SU(N)$ can be uniquely written as
\be
u = u_1(\alpha_1) u_2(\alpha_2) \cdots u_{N-1}(\alpha_{N-1})\,,\quad u_i \in
U(1)_i\,.
\ee
Moreover, the center $Z(SU(N))$ is contained in the last group and is generated by $\alpha_{N-1}=2\pi/N$.
One can therefore introduce $\overline{U(1)}_{N-1} = U(1)_{N-1}/Z(SU(N))$, which can be parametrized as
\be
\overline{U(1)}_{N-1} = \alpha_{N-1}\left(-{N-1 \over N}, \, {1\over N},\, \dots,\, {1 \over N}\right)\,,
\ee
where $\alpha_{N-1} \in [0,2\pi)$.
Since $Z(SU(N))\cap U(1)_j=1$ for $j=1,\ldots,N-2$, we identify these subgroups with their homomorphic images and omit the bar.
Thus the homomorphic image $T$ of $T_0$ can be written as
\be
T = U(1)_1\times U(1)_2 \times \cdots \times \overline{U(1)}_{N-1}\, . \label{maximal-torus-PSUN}
\ee

\subsection{Identifying symmetries of the potential}

Next we study which subgroups of the maximal torus $T$ can be realizable in the scalar sector of NHDM.

We start from the most general $T$-symmetric potential:
\be
V(T) = - \sum_a m_a^2(\phi_a^\dagger \phi_a) + \sum_{a,b} \lambda_{ab} (\phi_a^\dagger \phi_a)(\phi_b^\dagger \phi_b)
+ \sum_{a \not = b} \lambda'_{ab} (\phi_a^\dagger \phi_b)(\phi_b^\dagger \phi_a)\,,\label{Tsymmetric}
\ee
where among the $N(N-1)/2$ terms in the last sum there are only $2N-3$ algebraically independent ones,
\cite{Ivanov1}. Each term in this potential transforms trivially under the entire $T$. The important fact now is
that a sufficiently general potential of this form has no other unitary symmetry. In fact, an even stronger statement
is true:
\begin{theorem}\label{prop-T-symmetric}
Consider potential $V(T)$, (\ref{Tsymmetric}), defined by the set of coefficients $m_a^2$, $\lambda_{ab}$, $\lambda'_{ab}$.
\begin{enumerate}
\item
There exist coefficients such that the only unitary symmetries
of the potential are the phase rotations from $T$.
\item
If one constructs $V = V(T) + V_1$ by adding further terms that were absent in (\ref{Tsymmetric}),
then for any $V_1$ there still exist coefficients of $V(T)$ such that any unitary symmetry of $V$
belongs to $T$.
\end{enumerate}
\end{theorem}

\begin{proof}
The first assertion is a particular case of the second, which we now prove.
Consider a generic potential $V$ in the space of bilinears (\ref{Vadjoint}).
Any unitary transformation of the Higgs families corresponds to an orthogonal
transformation in the $r_i$-space.
Therefore, any unitary symmetry of the potential must leave invariant each of the five terms in (\ref{Vadjoint}).
Consider, specifically, the second and the forth terms:
\be
V = - Y_{ab}(\phi^\dagger_a \phi_b) + Z_{ab}(\phi^\dagger_a \phi_b) \sum_c (\phi^\dagger_c \phi_c) + \dots\,,\label{twoterms}
\ee
where we omitted the remaining terms and introduced the tensor
\be
Z_{ab} = {2 \over N}\sum_d Z_{abdd}-{2 \over N^2}\delta_{ab}\sum_{c,d}Z_{ccdd}\,,
\ee
which is basically the $Z^{(2)}_{ab}$ of \cite{Davidson} made traceless.
Both $Y_{ab}$ and $Z_{ab}$ are hermitean matrices, and their diagonal values can be arbitrarily changed
by changing $m_{a}^2$ and $\lambda_{ab}$ in (\ref{Tsymmetric}).
In particular, one can adjust them so that the spectra of $Y_{ab}$ and $Z_{ab}$ are non-degenerate.
Their eigenvectors can be then used to construct two orthonormal bases in the space of Higgs doublets,
$\{e^{(Y)}_i\}$ and $\{e^{(Z)}_i\}$, respectively ($i = 1,\dots, N$).
Let us also introduce the canonical basis $\{e^{(0)}_i\}$ consisting of vectors of the form $(0,\dots,1,\dots,0)$,
so that the maximal torus $T$ consists of phase rotations of these vectors.

The symmetry group of the first term in (\ref{twoterms}) is the maximal torus $T_Y$ of phase rotations
of individual eigenvectors $e^{(Y)}_i$.
The symmetry group of the second term in (\ref{twoterms}) is the maximal torus $T_Z$ of phase rotations
of individual eigenvectors $e^{(Z)}_i$.
The symmetry group of (\ref{twoterms}) is, therefore, $T_Y \cap T_Z$, the phase rotations of common eigenvectors
of the two matrices (and the symmetry group of the entire $V$ is its subgroup).

Take an eigenvector of $Y_{ab}$. Its orientation can be changed by readjusting coefficients $m_a^2$
if and only if it does not belong to the canonical basis. The same applies to $Z_{ab}$.
Therefore, even if $Y_{ab}$ and $Z_{ab}$ had common eigenvectors not belonging to the canonical basis,
one could avoid this coincidence by readjusting the coefficients $m_a^2$ and $\lambda_{ab}$.
After this procedure, the only common eigenvectors of $Y_{ab}$ and $Z_{ab}$ will be those
belonging to the canonical basis.
Therefore, the symmetry group of (\ref{twoterms}), and consequently, of the full potential will be a subgroup
of $T$.
\end{proof}

The theorem just proved guarantees that when we start from the $T$-symmetric potential (\ref{Tsymmetric})
and add more terms, we will never generate any new unitary symmetry that was not already present in $T$.
This is the crucial step in proving that the groups described below are realizable.
Our task now is to find which subgroups of $T$ can be obtained in this way.

Consider a bilinear $\phi_a^\dagger \phi_b$, $a\not = b$. It gains a phase change under $T$ (\ref{maximal-torus-PSUN})
which linearly depends on the angles $\alpha_i$:
\be
\phi_a^\dagger \phi_b \to \exp[i(p_{ab}\alpha_1 + q_{ab}\alpha_2 + \dots + t_{ab}\alpha_{N-1})]\cdot \phi_a^\dagger \phi_b\label{pq-generic-NHDM}
\ee
with some integer coefficients $p_{ab},\, q_{ab},\, \dots,\, t_{ab}$.
Note that all coefficients are antisymmetric in their indices: $p_{ba}= - p_{ab}$, etc.
These coefficients can be represented by real antisymmetric matrices with integer values, or graphically,
as labels of the edges of $N-1$ {\em oriented graphs}, one for each $U(1)$ group.
Each such graph has $N$ vertices, corresponding to doublets $\phi_a$; all vertices are joined with arbitrarily oriented edged,
orientation indicated by an arrow. An edge oriented from $\phi_b$ towards $\phi_a$ (edge $b \to a$) is associated with
the bilinear $\phi_a^\dagger\phi_b$ and is labeled by $p_{ba}$ in the first graph, $q_{ba}$ in the second graph, etc.
Examples of these graphs will be shown later for 3HDM and 4HDM.

The Higgs potential is a sum of monomial terms which are linear or quadratic in $\phi_a^\dagger \phi_b$.
Consider one such term and calculate its coefficients $p,\dots, t$.
Let us first focus on how this term depends on any single $U(1)_i$ subgroup of $T$.
There are two possibilities depending on the value of the $i$-th coefficient:
\begin{itemize}
\item
If the coefficient $k$ in front of $\alpha_i$ is zero, this terms is $U(1)_i$-symmetric.
\item
If the coefficient $k \not = 0$, then this term is symmetric under the $\Z_{|k|}$ subgroup of $U(1)_i$-group
generated by phase rotations by $2\pi/|k|$.
\end{itemize}
However, even if a given monomial happens to have finite symmetry groups with respect to each single $U(1)_i$,
its symmetry group under the entire $T$ is still continuous: for any set of coefficients one can adjust
angles $\alpha_i$ in such a way that $p_{ab}\alpha_1 + \dots + t_{ab}\alpha_{N-1} = 0$.
Therefore, when studying symmetries of a given term or a sum of terms, we cannot limit ourselves
to individual $U(1)_i$ groups but must consider the full maximal torus.

The strategy presented below guarantees that we find all possible realizable subgroups of the maximal torus $T$, both finite and infinite.

Consider a Higgs potential $V$ which, in addition to the $T$-symmetric part (\ref{Tsymmetric}) contains $k\ge 1$ additional terms,
with coefficients $p_1,\,q_1,\,\dots t_1$ to $p_k,\,q_k,\,\dots t_k$.
This potential defines the following $(N-1)\times k$ matrix of coefficients:
\be
\label{XV}
X(V) = \left(\begin{array}{cccc}
p_1 & q_1 & \cdots & t_1\\
p_2 & q_2 & \cdots & t_2\\
\vdots & \vdots && \vdots \\
p_k & q_k & \cdots & t_k
\end{array}
\right) = 
\left(\begin{array}{ccc}
m_{1,1} & \cdots & m_{1,N-1}\\[2.5mm]
\vdots & \vdots & \vdots \\[2.5mm]
m_{k,1} &  \cdots & m_{k,N-1}
\end{array}
\right)\,.
\ee
Here the second form of the matrix agrees with the notation of (\ref{systemUN}).
The symmetry group of this potential can be derived from the set of non-trivial solutions for $\alpha_i$ of the following equations:
\be
X(V) \left(\begin{array}{c} \alpha_1 \\ \vdots \\ \alpha_{N-1} \end{array}\right)
= \left(\begin{array}{c} 2\pi n_1 \\ \vdots \\2 \pi n_{k} \end{array}\right)\,,\label{XVeq}
\ee
There are two major possibilities depending on the rank of this matrix.
\begin{itemize}
\item
{\bf Finite symmetry group.}
If rank$X(V)=N-1$, then there is no non-trivial solution of the equation (\ref{XVeq}) with the trivial right-hand side (i.e. all $n_i = 0$).
Instead, there exists a unique solution for any non-trivial set of $n_i$, and all such solutions form the finite group of phase rotations
of the given potential.

To find which symmetry groups can be obtained in this way,
we take exactly $N-1$ monomials, so that the matrix $X(V)$ becomes a square matrix with a non-zero determinant.
It is known that any square matrix with integer entries can be diagonalized by a sequence of the following
operations: swapping two rows or two columns,
adding a row (a column) to another row (column) and multiplying a row (a column) by $-1$.
After diagonalization, the matrix $X(V)$ becomes diag$(d_1,\dots, d_{N-1})$, where $d_i$ are positive non-zero integers.
This matrix still defines the equation (\ref{XVeq}) for $\alpha_i'$, which are linear combinations of $\alpha_i$'s and
with $n_i^\prime \in \Z$. Therefore, the finite symmetry group of this matrix is $\Z_{d_1} \times \cdots \times \Z_{d_{N-1}}$
(where $\Z_1$ means no non-trivial symmetry).

Note also that each of the allowed manipulations conserves the absolute value of the determinant of $X(V)$.
Therefore, even before diagonalization one can calculate the order of the finite symmetry group as $|\det X(V)|$.

This derivation leads us to the strategy that identifies all finite subgroups of torus realizable as symmetry groups of the Higgs potential
in NHDM: write down all possible monomials with $N$-doublets, consider all possible subsets with exactly $N-1$ distinct monomials,
construct the matrix $X$ for this subset and find its symmetry group following the above scheme.
Although this strategy is far from being optimal, its algorithmic nature allows it to be easily implemented in a machine code.
\item
{\bf Continuous symmetry group.}
If $\mathrm{rank} X(V) < N-1$, so that $D=N-1-\mathrm{rank} X(V) > 0$, then there exists a $D$-dimensional subspace
in the space of angles $\alpha_i$, which solves the equation (\ref{XVeq}) for the trivial right-hand side.
One can then focus on the orthogonal complement of this subspace, where no non-trivial solution
of the homogeneous equation is possible, and repeat the above strategy
to find the finite symmetry group $G_D$ in this subspace. The symmetry group of the potential
is then $[U(1)]^D \times G_D$.
\end{itemize}

\subsection{Identifying symmetries using the dual group}

The strategy just exposed can be put on a firmer algebraic ground with the classic formalism of the {\em Pontryagin duality}, \cite{Pontr}.

Let $T$ be a subgroup of all diagonal matrices of $PSU(N)$, i.e., $T$ is a maximal torus, and consider the decomposition $T=U_1\times\ldots\times U_n$, where $n=N-1$. By construction, $U_i=\{u_i(\alpha_i)\mid \alpha_i\in [0,2\pi)\}$. By $K$ we denote the factor group $\mathbb{R}/\mathbb{Z}$ (we use additive notation for the operation in this group) and, given an abelian group $A$ the homomorphism $\chi:A\rightarrow K$ is called a {\em character} of $A$. The set of all characters of $A$ become an abelian group under the operation given by $(\chi_1+\chi_2)(a)=\chi_1(a)+\chi_2(a)$ and we denote this group by $X(A)$. 
This is the dual group of $A$. If $X(A)$ is found, then $A$ can be recovered by considering $X(X(A))$ the dual group of $X(A)$.

Consider $\chi_i\in X(T)$ defined on generators by 
$$
\chi_i(u_j(\alpha_j))=\left\{
\begin{array}{rl}
\frac{\alpha_j}{2\pi},&\text{ if }i=j;\\
0,&\text{ if }i\neq j.                              
\end{array}
\right.
$$
Clearly the cyclic group $X_i:=\langle \chi_i\rangle$ is isomorphic to $X(U_i)$ and $X(T)=X_1\times\ldots\times X_n$, i.e., $\chi_1,\ldots,\chi_n$ are free generators of a free abelian group $X(T)$. Now consider a potential $V$, denote $X(V)$ by $M$ (recall that $M$ is a $k\times n$-matrix with integer items) and denote the items of $M$ by $m_{i,j}$. An element $x=u_1(\alpha_1)\cdot\ldots\cdot u_n(\alpha_n)\in T$ lies in $\mathrm{Aut}(V)$ if and only if (\ref{XVeq}) holds.
That is, $x\in \mathrm{Aut}(V)$ if and only if for each $i$  we have $\xi_i(x)=0$, where the character $\xi_i$ is defined by $\xi_i=\sum_{j=1}^n m_{i,j}\chi_j$. 

Consider a subgroup $H=\langle \xi_1,\ldots,\xi_k\rangle$ of $X(T)$. The above arguments imply that $H$ is the annulator of $\mathrm{Aut}(V)$, 
and that $X(T)/H\simeq X(\mathrm{Aut}(V))$. If we diagonalize $M$ by the standard transformations we obtain the following matrix 
\begin{equation*}
D=\left(
\begin{array}{ccccc}
d_1&0&0&\ldots&0\\
0&d_2&0&\ldots&0\\
\ldots&\ldots&\ldots&\ldots\\
0&\ldots &d_k&\ldots&0\\
\end{array}
\right),
\end{equation*}
where $d_{i-1}$ divides $d_i$ for each $i$. Now \cite[Theorem~8.1.1]{Karg} implies that  
$$X(T)/H\simeq \mathbb{Z}_{d_1}\times\ldots\times\mathbb{Z}_{d_k}\times\mathbb{Z}^{n-k},$$
where $\mathbb{Z}_{d_i}$ is a cyclic group of order $d_i$ and $\mathbb{Z}^{n-k}$ is a direct product of $n-k$ copies of $\mathbb{Z}$. 
In view of the Pontryagin duality between groups and characters, we obtain that 
$$\mathrm{Aut}(V)\simeq \mathbb{Z}_{d_1}\times\ldots\times\mathbb{Z}_{d_k}\times U\,,$$
where $U$ is an $(n-k)$-dimensional torus.

\subsection{Antiunitary transformations}\label{subsection-antiunitary-general}

Having found the list of realizable subgroups of torus $T$, we can ask whether these groups
can be extended to larger abelian groups that would include not only unitary but also antiunitary transformations.
Here we describe the strategy that solves this problem.

As described above, one can define the action of the $CP$-transformation $J$ on $PSU(N)$
given by (\ref{JactiononUN}): $J\cdot U \cdot J = U^*$.
This action can be restricted to the maximal torus $T$ (\ref{maximal-torus-PSUN}) and even further, to any subgroup $A \subset T$
(indeed, if $a \in A$, then $a^* \in A$). However, $J$ does not commute with a generic element of $T$,
therefore the group $\langle A, J\rangle$ is not, in general, abelian.

In order to embed $A$ in a larger abelian group, we must find
an antiunitary transformation $J' = b\cdot J$ which commutes with any element $a \in A$:
\be
J'a(J')^{-1} = a\quad \Leftrightarrow \quad bJaJb^{-1}=a \quad \Leftrightarrow \quad ba^{-1}b^{-1}=a \quad \Leftrightarrow\quad  b=aba.\label{Jprime}
\ee
The last expression here is a linear matrix equation for the matrix $b$ at any given $a$.
Since $a$ is diagonal, $(aba)_{ij} = a_{ii}a_{jj}b_{ij}$, so whenever $a_{ii}a_{jj} \not = 1$,
one must set $b_{ij}$ to zero.

If at least one matrix $b$ satisfying (\ref{Jprime}) is found, all the other matrices can be constructed with the help of the last form of this equation.
Suppose $b'=xb$ also satisfies this equation; then, $axa^{-1}=x$. Therefore, $x$ can be any unitary matrix from $PSU(N)$ commuting with $a$,
that is, $x$ centralizes the chosen abelian group $A$.

In summary, the strategy for embedding a given abelian group $A \subset T$ into a larger abelian group $A_{CP}$
which includes antiunitary transformations proceeds in five steps:
\begin{itemize}
\item
Find one $b$ which solves the matrix equation (\ref{Jprime}) for all $a \in A$.
\item
Find all $x \in PSU(N)$ which commute with each element of $A$, that is, construct the centralizer of $A$ in $PSU(N)$.
\item
Find restrictions placed on $x$'s by the requirement that product of two different antiunitary transformations $(x_1J')$ and $(x_2 J')$ belongs to $A$.
This procedure can result in several distinct groups $A_{CP}$.
\item
Find which among the $A$-symmetric terms in the potential are also symmetric under some of $A_{CP}$;
drop terms which are not.
\item
Check that the resulting collection of terms is not symmetric under a larger group of {\em unitary} transformations,
that is, $A$ is still a realizable symmetry group of this collection of terms.
\end{itemize}

As an exercise, let us apply this strategy to the full maximal torus $T$.
A generic $a \in T$ acting on a doublet $\phi_i$ generates a non-trivial phase rotation
$\psi_i(\alpha_1,\dots,\alpha_{N-1})$ which can be reconstructed from (\ref{groupsUi}).
The equation (\ref{Jprime}) then becomes
\be
e^{i(\psi_i+\psi_j)}b_{ij} = b_{ij}\,.
\ee
Since $\psi_i + \psi_j \not = 0$ for any $i, j$, the only solution to this equation is $b_{ij}=0$
for all $i$ and $j$.
This means that there is no $J'$ that would commute with every element of the torus $T$.
Thus, $T$ cannot be embedded into a larger abelian group $T_{CP}$ that would include antiunitary transformations.

In order to avoid possible confusion, we note that in fact $J$ {\em is} a symmetry of the $T$-symmetric potential.
However, the symmetry group, generated by $T$ and $J$, is a semidirect product $T\rtimes J$, and it is nonabelian.

\section{Examples of 3HDM and 4HDM}\label{section-3HDM-4HDM}

In this Section we illustrate the general strategy presented above with the particular cases of 3HDM and 4HDM.
For 3HDM we give the full list of abelian groups realizable as symmetry groups of the Higgs potential
and show explicit examples of such potentials; for 4HDM we do it only for the finite groups of unitary transformations.

\subsection{Maximal torus for 3HDM and its subgroups}

In the 3HDM the representative maximal torus $T \subset PSU(3)$ is parametrized as
\be
T = U(1)_1\times U(1)_2\,,\quad U(1)_1 = \alpha(-1,\,1,\,0)\,,\quad U(1)_2 = \beta\left(-{2 \over 3},\, {1 \over 3},\, {1 \over 3}\right)\,,
\label{3HDM-maximaltorus}
\ee
where $\alpha,\beta \in [0,2\pi)$.
The most general Higgs potential symmetric under $T$ is
\bea
\label{3HDMpotential-Tsymmetric}
V &=& - m_1^2 (\phi_1^\dagger \phi_1) - m_2^2 (\phi_2^\dagger \phi_2) - m_3^2 (\phi_3^\dagger \phi_3)\\
&+& \lambda_{11} (\phi_1^\dagger \phi_1)^2 + \lambda_{22} (\phi_2^\dagger \phi_2)^2 + \lambda_{33} (\phi_3^\dagger \phi_3)^2\nonumber\\
&+& \lambda_{12} (\phi_1^\dagger \phi_1) (\phi_2^\dagger \phi_2)
+ \lambda_{23} (\phi_2^\dagger \phi_2) (\phi_3^\dagger \phi_3)
+ \lambda_{13} (\phi_3^\dagger \phi_3) (\phi_1^\dagger \phi_1)
\nonumber\\
&+& \lambda'_{12} |\phi_1^\dagger \phi_2|^2 + \lambda'_{23} |\phi_2^\dagger \phi_3|^2 + \lambda'_{13} |\phi_3^\dagger \phi_1|^2\,.\nonumber
\eea
There are six bilinear combinations of doublets transforming non-trivially under $T$.
The way they transform under $U(1)_1$ and $U(1)_2$ is described by two coefficients $p$ and $q$ defined in (\ref{pq-generic-NHDM}).
Here is the list of these coefficients for the three independent bilinears:
\be
\begin{tabular}{c | cc}
& $p$ & $q$ \\
\hline
$(\phi_2^\dagger \phi_1)$ & $-2$ & $-1$ \\
$(\phi_3^\dagger \phi_2)$ & $1$ & $0$ \\
$(\phi_1^\dagger \phi_3)$ & $1$ & $1$
\end{tabular}
\label{pq-explicit-3HDM}
\ee
For the conjugate bilinears, the coefficients $p$ and $q$ carry the opposite signs with respect to (\ref{pq-explicit-3HDM}).
These coefficients are shown graphically
as labels of the edges of two oriented graphs shown in Fig.~\ref{graph3HDM}.

\begin{figure} [ht]
\centering
\includegraphics[height=3cm]{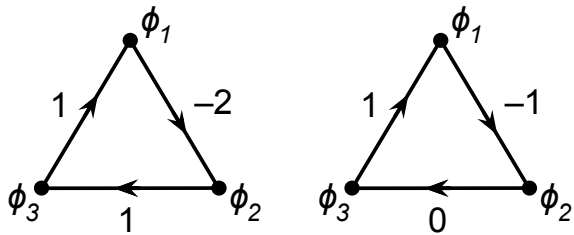}
\caption{Coefficients $p$ and $q$ as labels of triangles representing bilinears
in 3HDM}
\label{graph3HDM}
\end{figure}

Before applying the general strategy described in the previous Section, let us
check which finite symmetry groups can arise as subgroups of either $U(1)_1$ or $U(1)_2$ individually.

Consider first $U(1)_1$. The coefficient $p$ of any monomial can be obtained either by picking the labels from the first graph
in Fig.~\ref{graph3HDM} directly or by summing two such labels, multiplied by $-1$ when needed.
In this way we can obtain any $|p|$ from $0$ to $4$. For a monomial with $p=0$, the symmetry group is the entire $U(1)_1$.
For $|p|=1$, there is no non-trivial symmetry. For $|p| = 2,3,4$, we obtain the cyclic group $\Z_{|p|}$.
In each case it is straightforward to construct the monomials with a given symmetry;
for example, $(\phi_1^\dagger \phi_3)(\phi_2^\dagger \phi_3)$ and its conjugate are $U(1)_1$-symmetric,
while $(\phi_1^\dagger \phi_2)(\phi_1^\dagger \phi_3)$,  $(\phi_1^\dagger \phi_2)(\phi_3^\dagger \phi_2)$
and their conjugates are symmetric under the $\Z_3$-group with a generator
\be
a = (\omega,\,1,\,\omega^{-1})\,,\quad \omega \equiv \exp\left(2\pi i/3\right)\,.\label{Z3generator}
\ee
In the case of the group $U(1)_2$, the labels can sum up to $|q| = 0,1,2$. Thus, the largest finite group here is $\Z_2$.

As we mentioned in the previous Section, one cannot limit oneself to subgroups of individual $U(1)_i$ factors
or to direct products of such subgroups.
In order to find all realizable groups, one has to write the full list of possible monomials and then calculate the symmetry
group of all distinct pairs (for $N=3$) of monomials. For example, if the two monomials are
$v_1 = \lambda_1(\phi_1^\dagger \phi_2)(\phi_1^\dagger \phi_3)$ and $v_2 = \lambda_2(\phi_2^\dagger \phi_1)(\phi_2^\dagger \phi_3)$,
then the matrix $X(v_1+v_2)$ defined in (\ref{XV}) has form
\be
X(v_1+v_2) = \left(\begin{array}{cc} 3 & 2 \\ -3 & -1 \end{array}\right)\,.
\ee
It can be diagonalized by adding the first row to the second, and then subtracting the second row twice from the first:
\be
\left(\begin{array}{cc} 3 & 2 \\ -3 & -1 \end{array}\right) \to \left(\begin{array}{cc} 3 & 2 \\ 0 & 1 \end{array}\right)
\to \left(\begin{array}{cc} 3 & 0 \\ 0 & 1 \end{array}\right)\,.
\ee
The diagonal for of the matrix obtained indicates that the symmetry group of the potential is $\Z_3$. The solution of the equation
\be
X(v_1+v_2) \left(\begin{array}{c} \alpha \\ \beta \end{array}\right) = \left(\begin{array}{c} 2\pi n_1 \\ 2 \pi n_2 \end{array}\right)\label{Xv1v2eq}
\ee
yields $\alpha = 2\pi/3\cdot k$, $\beta=0$,
which implies the transformation matrix of the doublets (\ref{Z3generator}).

In 3HDM there are, up to complex conjugation, three bilinears and nine products of two bilinears transforming non-trivially under $T$.
It is a straightforward exercise to check all possible pairs of monomials; in fact, by using permutations of the doublets
the number of truly distinct cases is rather small. This procedure reveals just one additional finite group $\Z_2 \times \Z_2$,
which arises when at least two terms among $(\phi_1^\dagger\phi_2)^2$, $(\phi_2^\dagger\phi_3)^2$, $(\phi_3^\dagger\phi_1)^2$
are present. This group is simply the group of sign flips of individual doublets.

Thus, we arrive at the full list of subgroups of the maximal torus realizable as the symmetry groups of the Higgs potential in 3HDM:
\be
\Z_2,\quad \Z_3,\quad  \Z_4,\quad \Z_2\times \Z_2,\quad U(1),\quad U(1)\times \Z_2,\quad U(1)\times U(1)\,. \label{list3HDM}
\ee
Most of these groups were identified in \cite{Ferreira} in the search of ``simple'' symmetries of 3HDM scalar potential.
In that work a symmetry were characterized not by its group, as in our paper, but by providing
a single unitary transformation $S$ and then reconstructing the pattern in the Higgs potential which arises
after requiring that it is $S$-symmetric.
In certain cases the authors of  \cite{Ferreira} found that the potential is symmetric under a larger group than $\langle S\rangle$,
in accordance with the notion of realizable symmetry discussed above.

The explicit correspondence between the seven symmetries $S_1,\dots, S_7$ of \cite{Ferreira} and the list (\ref{list3HDM})
is the following:
\bea
&&S_1 \to \Z_2\,,\quad S_2 \to U(1) \ \mbox{realized as}\ U(1)_2\,,\quad S_3 \to \Z_3\,,\label{correspondence}\\
&&S_4 \to \Z_4\,,\quad S_5 \to U(1) \ \mbox{realized as}\ U(1)_1\,,\quad S_6 \to U(1)\times \Z_2\,,
\quad S_7 \to U(1)\times U(1)\,.\nonumber
\eea
In addition to these symmetries, our list (\ref{list3HDM}) contains one more group $\Z_2\times \Z_2$,
which was not found in \cite{Ferreira} because it does not correspond to a ``simple'' symmetry.

Let us also explicitly write the potentials which are symmetric under each group in (\ref{list3HDM}).
\begin{itemize}
\item
${\bf U(1)\times U(1) = T}$. The most general $T$-symmetric potential of 3HDM is given by (\ref{3HDMpotential-Tsymmetric}).
\item
${\bf U(1)}$. This group can be realized in two non-equivalent ways, which are conjugate either to $U(1)_1$ or $U(1)_2$ in (\ref{3HDM-maximaltorus}).
The distinction between the two realizations lies in the eigenspace decomposition:
a $U(1)_2$-type group has a two-(complex)-dimensional space of eigenvectors, while eigenspaces
of a $U(1)_1$-type group are all one-dimensional. This leaves more freedom in constructing a $U(1)_2$-symmetric potential
than a $U(1)_1$-symmetric one.
Specifically, the general $U(1)_1$-invariant potential contains, in addition to (\ref{3HDMpotential-Tsymmetric}), the following terms:
\be
\lambda_{1323}(\phi_1^\dagger\phi_3)(\phi_2^\dagger\phi_3) + h.c.\label{U11potential}
\ee
while the general $U(1)_2$-invariant potential contains
\bea
&&-m_{23}^2(\phi_2^\dagger\phi_3) + \left[\lambda_{1123} (\phi_1^\dagger\phi_1) + \lambda_{2223} (\phi_2^\dagger\phi_2)
+ \lambda_{3323} (\phi_3^\dagger\phi_3)\right](\phi_2^\dagger\phi_3)\nonumber\\
&&+ \lambda_{2323}(\phi_2^\dagger\phi_3)^2 + \lambda_{2323}(\phi_2^\dagger\phi_1)(\phi_1^\dagger\phi_3) + h.c. \label{U12potential}
\eea
It must be stressed that these potentials are written for the specific convention of groups $U(1)_1$ and $U(1)_2$ used in
(\ref{3HDM-maximaltorus}). This convention reflected a specific basis used to describe the structure of the torus $T$.
In other bases, the explicit terms symmetric under a $U(1)_1$ or $U(1)_2$-type groups will look differently.
For example, the term $(\phi_2^\dagger\phi_1)(\phi_3^\dagger\phi_1)$ is symmetric under a $U(1)_1$-type
transformation with phases $(0,\alpha,-\alpha)$.
To this end, we note that reconstructing the symmetry group of any given potential is a separate and difficult task
which is not addressed in the present paper (here we only classify realizable symmetries and give examples
of their realizations).
\item
${\bf U(1) \times \Z_2}$.
Looking back to (\ref{U11potential}), one might be tempted to think that the true symmetry group of this term is not $U(1)$ but $U(1)\times\Z_2$,
because this term is also invariant under sign flip of $\phi_3$. However, inside $SU(3)$, a transformation with phases $(0,0,\pi)$
is equivalent to $(-\pi,\pi,0)$ which in already included in $U(1)_1$.

A potential whose true symmetry group is $U(1) \times \Z_2$ is given by (\ref{U12potential}) with all coefficients set to zero
except for $\lambda_{2323}$. The corresponding term, $(\phi_2^\dagger\phi_3)^2$, is symmetric not only under the full $U(1)_2$,
but also under $(-\pi,\pi,0)$, which generates a $\Z_2$ subgroups inside $U(1)_1$.
\item
${\bf \Z_4}$.
The potential symmetric under $\Z_4$ contains, in addition to (\ref{3HDMpotential-Tsymmetric}), the following terms:
\be
\lambda_{1323}(\phi_1^\dagger\phi_3)(\phi_2^\dagger\phi_3) + \lambda_{1212}(\phi_1^\dagger\phi_2)^2 + h.c.
\ee
In accordance with the general discussion, it must contains $N-1=2$ distinct terms, as required for a finite symmetry group.
Again, this set of terms is specific for the particular choice of the $(U(1)_1,U(1)_2)$-basis on the maximal torus.
\item
${\bf \Z_3}$.
The potential symmetric under $\Z_3$ contains
\be
\lambda_{1232}(\phi_1^\dagger\phi_2)(\phi_3^\dagger\phi_2) +
\lambda_{2313}(\phi_2^\dagger\phi_3)(\phi_1^\dagger\phi_3) +
\lambda_{3121}(\phi_3^\dagger\phi_1)(\phi_2^\dagger\phi_1) + h.c.
\ee
In fact, any pair of the three terms is already sufficient to define a $\Z_3$-symmetric potential.
Note also that different $(U(1)_1,U(1)_2)$-basis on the maximal torus lead to the same $\Z_3$-group, because
inside $PSU(3)$ the following groups are equal:
\be
\langle\left(1,\omega,\omega^2\right)\rangle \simeq
\langle\left(\omega,\omega^2,1\right)\rangle \simeq
\langle\left(\omega^2,1,\omega\right)\rangle\,,\quad \omega \equiv \exp(2\pi i/3)\,.
\ee
\item
${\bf \Z_2 \times \Z_2}$.
This group can be realized simply as a group of independent sign flips of the three doublets (up to the overall sign flip).
Every term in the potential must contain each doublet in pairs. In addition to (\ref{3HDMpotential-Tsymmetric}),
the $\Z_2 \times \Z_2$-symmetric potential can contain
\be
\lambda_{1212}(\phi_1^\dagger\phi_2)^2 +
\lambda_{2323}(\phi_2^\dagger\phi_3)^2 +
\lambda_{3131}(\phi_3^\dagger\phi_1)^2 + h.c.
\ee
\item
${\bf \Z_2}$.
This group can be realized, for example, as a group of sign flips of the third doublet.
The potential can contain any term where $\phi_3$ appears in pairs.

\end{itemize}

\subsection{The $\Z_3 \times \Z_3$-group}

The only finite abelian group that is not contained in any maximal torus in $PSU(3)$ is $\Z_3 \times \Z_3$.
Although there are many such groups inside $PSU(3)$, all of them are conjugate to each other. Thus, only one
representative case can be considered, and we describe it with the following two generators
\be
a = \left(\begin{array}{ccc} 1 & 0 & 0 \\ 0 & \omega & 0 \\ 0 & 0 & \omega^2 \end{array}\right),\quad
b = \left(\begin{array}{ccc} 0 & 1 & 0 \\ 0 & 0 & 1 \\ 1 & 0 & 0 \end{array}\right),\quad \omega = \exp\left({2\pi i \over 3}\right)\,.
\ee
A generic potential that stays invariant under this group of transformations is
\bea
V & = &  - m^2 \left[(\phi_1^\dagger \phi_1)+ (\phi_2^\dagger \phi_2)+(\phi_3^\dagger \phi_3)\right]
+ \lambda_0 \left[(\phi_1^\dagger \phi_1)+ (\phi_2^\dagger \phi_2)+(\phi_3^\dagger \phi_3)\right]^2 \nonumber\\
&&+ \lambda_1 \left[(\phi_1^\dagger \phi_1)^2+ (\phi_2^\dagger \phi_2)^2+(\phi_3^\dagger \phi_3)^2\right]
+ \lambda_2 \left[|\phi_1^\dagger \phi_2|^2 + |\phi_2^\dagger \phi_3|^2 + |\phi_3^\dagger \phi_1|^2\right] \nonumber\\
&&+ \lambda_3 \left[(\phi_1^\dagger \phi_2)(\phi_1^\dagger \phi_3) + (\phi_2^\dagger \phi_3)(\phi_2^\dagger \phi_1) + (\phi_3^\dagger \phi_1)(\phi_3^\dagger \phi_2)\right]
+ h.c.\label{VZ3Z3}
\eea
with real $m^2$, $\lambda_0$, $\lambda_1$, $\lambda_2$ and complex $\lambda_3$.
It is also interesting to note that this group describes an abelian frustrated symmetry of the 3HDM potential,
see details in \cite{frustrated}.

The potential (\ref{VZ3Z3}) is not symmetric under any continuous Higgs-family transformation,
which can be proved, for example, using the adjoint representation of the bilinears described in \cite{Ivanov1}
and briefly outlined in section~\ref{subsection-potential}.
However it is symmetric under the exchange of any two doublets,
e.g. $\phi_1 \leftrightarrow \phi_2$. Together with $b$, such transformations generate the group $S_3$ of permutations of the three doublets.
Its action on elements of the $\Z_3$ generated by $a$ is also naturally defined,
therefore, the potential (\ref{VZ3Z3}) is automatically symmetric under the group $(\Z_3 \times \Z_3)\rtimes \Z_2$, which is non-abelian.

We can conclude that the symmetry group $\Z_3 \times \Z_3$ is not realizable for 3HDM.

\subsection{Including antiunitary symmetries in 3HDM}

In Section \ref{subsection-antiunitary-general} we described the strategy of embedding an abelian group of unitary transformations
into a larger abelian group
which includes anti-unitary transformations. Applying this strategy to 3HDM for all groups from the list (\ref{list3HDM}),
we obtain the following additional realizable abelian groups:
\be
\Z_2^*\,,\quad \Z_2\times \Z_2^*\,, \quad \Z_2\times \Z_2\times \Z_2^*\,, \quad \Z_4^*\,,
\ee
where the asterisk indicates that the generator of the corresponding cyclic group is an anti-unitary transformation.
Details of the derivation are presented in Appendix~\ref{appendix-antiunitary}. Here we just briefly explain why groups such as $\Z_6^*$, $\Z_8^*$,
$U(1)\times \Z_2^*$ do not appear in this list.
When we search for anti-unitary transformations commuting wth the selected subgroup of $T$,
we can indeed construct such groups at the price of imposing certain restrictions on the coefficients of the $T$-symmetric part
(\ref{3HDMpotential-Tsymmetric}). This makes the potential symmetric under a larger group of {\em unitary} transformations, which
is nonabelian. According to our definition, this means that the groups such as $\Z_6^*$ are not realizable,
although they are subgroups of larger realizable non-abelian symmetry groups.

\subsection{Abelian symmetries of the 4HDM potential}

In the case of four Higgs doublets the representative maximal torus in $PSU(4)$ is $T = U(1)_1\times U(1)_2 \times U(1)_3$,
where
\be
U(1)_1  =  \alpha(-1,\, 1,\, 0,\, 0)\,,\quad
U(1)_2  =  \beta(-2,\, 1,\, 1,\, 0)\,,\quad
U(1)_3  =  \gamma\left(-{3 \over 4}, \, {1 \over 4},\, {1 \over 4},\, {1 \over 4}\right)\,.
\ee
The phase rotations of a generic bilinear combination of doublets under $T$ is characterized by three integers $p,\,q,\,r$,
\be
(\phi_a^\dagger \phi_b) \to \exp[i(p_{ab}\alpha + q_{ab}\beta + r_{ab}\gamma)](\phi_a^\dagger \phi_b)\,.
\ee
These coefficients can again be represented graphically as labels of edges of three simplices shown in Fig.~\ref{graph4HDM}.

\begin{figure} [ht]
\centering
\includegraphics[height=4cm]{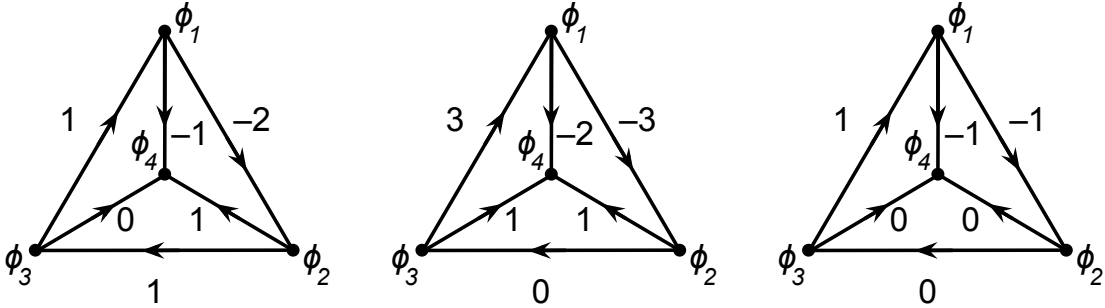}
\caption{Coefficients $p$, $q$, and $r$ as labels of triangles representing bilinears
in 4HDM}
\label{graph4HDM}
\end{figure}

With four doublets, one can construct (up to conjugation) 42 monomials transforming non-trivially under $T$.
When we search for realizable finite groups, we pick up all possible combinations of three distinct monomials $v_1$, $v_2$, $v_3$,
construct the $3 \times 3$ matrix $X(v_1+v_2+v_3)$, check that it is non-degenerate and then diagonalize it to obtain diag$(d_1,d_2,d_3)$.
The symmetry group is then $\Z_{d_1}\times \Z_{d_2}\times \Z_{d_3}$.
Although this brute force algorithm can be easily implemented in a machine code,
we can in fact apply the results of the next Section to find the list of all
realizable finite subgroups of the maximal torus in 4HDM:
\be
\Z_k\ \mbox{with}\ k = 2,\, \dots ,\, 8; \qquad
\Z_2\times \Z_k\ \mbox{with}\ k = 2,\,3,\,4; \qquad
\Z_2 \times \Z_2 \times \Z_2\,.\label{list4HDMfinite}
\ee
In short, these are all finite abelian groups of order $\le 8$.
Perhaps, the most surprising of these groups is $\Z_7$, because it does not appear as a realizable subgroup of
any single $U(1)_i$ factor (no two labels on Fig.~\ref{graph4HDM} sum up to 7).
An example of terms that possess this symmetry group is
\be
\lambda (\phi_1^\dagger \phi_3) (\phi_1^\dagger \phi_4) +
\lambda' (\phi_2^\dagger \phi_1) (\phi_2^\dagger \phi_4) +
\lambda'' (\phi_3^\dagger \phi_2) (\phi_3^\dagger \phi_4)\,,
\ee
and the $\Z_7$ group is generated by the transformation with phases
\be
a = {2\pi \over 28}\left(-9,\,-1,\,3,\,7\right)\,.
\ee
One can immediately check that $a^7$ lies in the center $Z(SU(4))$ and is, therefore, equivalent to the trivial
transformation.

Realizable continuous abelian symmetry groups as well as abelian groups
containing antiunitary tansformations can also be found following the general strategy.

\section{Abelian symmetries in general NHDM}\label{section-NHDM}

The algorithm described above can be used to find all abelian groups realizable as the symmetry groups of the Higgs potential
for any $N$. We do not yet have the full list of finite abelian groups for a generic $N$ presented in a compact form,
although we put forth a conjecture concerning this issue, see Conjecture~\ref{conjecture-list} below.
However, several strong results can be proved about the order and possible structure of finite realizable subgroups of the maximal torus.

Throughout this Section we will often use $n:=N-1$. Also, whenever we mention in this Section a finite abelian group
we actually imply a finite realizable subgroup of the maximal torus.

\subsection{Upper bound on the order of finite abelian groups}

It can be expected from the general construction that for any given $N$ there exists an upper bound on the order of
finite realizable subgroups of the maximal torus in NHDM. Indeed, in this Section we will prove the following theorem:
\begin{theorem}\label{theorem-order}
The exact upper bound on the order of the realizable finite subgroup of maximal torus in NHDM is
\be
|G| \le 2^{N-1}\,.
\ee
\end{theorem}

Before presenting the proof, let us first develop some convenient tools.
First, with the choice of the maximal torus (\ref{maximal-torus-PSUN}), we construct $n=N-1$ bilinears
$(\phi_1^\dagger \phi_{i+1})$, $i=1,\dots,n$.
The vectors of coefficients $a_i = (p_i,\, q_i,\, \dots,\, t_i)$ defined in (\ref{pq-generic-NHDM}) can be easily written:
\bea
a_1 &=& (2,3,4,\dots,n,1)\,,\nonumber\\
a_2 &=& (1,3,4,\dots,n,1)\,,\nonumber\\
a_3 &=& (1,2,4,\dots,n,1)\,,\nonumber\\
\vdots && \vdots \nonumber\\
a_n &=& (1,2,3,\dots,n-1,1)\,.
\eea
One can use these vectors to construct the $n\times n$ matrix $A$:
\be
A = \left(\begin{array}{c}
a_1\\
\vdots\\
a_n
\end{array}
\right)\,, \quad \det A = 1\,.\label{detaa}
\ee
From the unit determinant we can also conclude that after diagonalization the matrix $A$ becomes the identity matrix.

Now consider a bilinear $(\phi_i^\dagger \phi_j)$ with $i,j \not = 1$; its vector of coefficients can be represented as $a_{j-1} - a_{i-1}$.
More generally, for any monomial $(\phi_i^\dagger \phi_j)(\phi_k^\dagger \phi_m)$ with any $i,j,k,m$, the vector of coefficients
has the form $a_{j-1} - a_{i-1} + a_{m-1} - a_{k-1}$, where $a_0$ is understood as zero.
This means that the vector of coefficients of any monomial can be represented as a linear combination of $a$'s
with coefficients $0$, $\pm 1$ and $\pm 2$.
Since the $X$-matrix (\ref{XV}) is constructed from $n$ such vectors,  we can represent it as
\be
X_{ik} = c_{ij}A_{jk}\,.\label{XcA}
\ee
The square $n \times n$ matrix $c_{ij}$ can contain rows only of the following nine types (up to permutation and overall sign change):
\bea
\mbox{type }1: && (1,\,0,\,\cdots,\,0)\,,\nonumber\\
\mbox{type }2: && (2,\,0,\,\cdots,\,0)\,,\nonumber\\
\mbox{type }3: && (1,\,1,\,0,\,\cdots,\,0)\,,\nonumber\\
\mbox{type }4: && (1,\,-1,\,0,\,\cdots,\,0)\,,\nonumber\\
\mbox{type }5: && (2,\,-1,\,0,\,\cdots,\,0)\,,\label{nine-types}\\
\mbox{type }6: && (1,\, 1,\,-1,\,0,\,\cdots,\,0)\,,\nonumber\\
\mbox{type }7: && (2,\, -2,\,0,\,\cdots,\,0)\,,\nonumber\\
\mbox{type }8: && (2,\, -1,\,-1,\,0,\,\cdots,\,0)\,,\nonumber\\
\mbox{type }9: && (1,\, 1,\,-1,\,-1,\, 0,\,\cdots,\,0)\,.\nonumber
\eea
It follows from (\ref{XcA}) and (\ref{detaa}) that $\det X = \det c\cdot \det A = \det c$.
Therefore, order of any finite group is given by the module of determinant of $c$: $|G| = |\det c|$.

Let us also note two properties of the strings of type 1--9. Take any such string of length $n$, $x_{(n)} = (x_1,\,x_2,\,\dots,\,x_n)$,
which is obtained from (\ref{nine-types}) by an arbitrary permutation and possibly an overall sign flip.
Then any its substring $x_{(n-1)} = (x_1,\,\dots,\,x_{k-1},\, x_{k+1},\,\dots,\,x_n)$ obtained by removing
an arbitrary element $x_k$ is also of type 1--9.
Moreover, the element removed can be added at any place, and still the string
$x'_{(n-1)} = (x_1,\,\dots,\,x_m+x_k,\,\dots,\,x_{k-1},\, x_{k+1},\,\dots,\,x_n)$
remains of type 1--9.
Both properties can be proved by direct inspection of all the strings.

Now we are ready to prove Theorem~\ref{theorem-order}.

\begin{proof}
We prove by induction. Suppose that for any $(n-1)\times (n-1)$ square matrix $D_{n-1}$
whose rows are strings of type 1--9, its determinant $d_{n-1} = \det D_{n-1}$ is limited by
$|d_{n-1}| \le 2^{n-1}$. Take now any $n\times n$ matrix $D_{n}$ constructed from the same
family of strings and compute its determinant $d_{n}$ by minor expansion over the first row,
with $d^{(1)}_{n-1}$, $d^{(2)}_{n-1}$, $\dots$, being the relevant minors.
The procedure then depends on what type the first row is.
\begin{enumerate}
\item
If the first row is of type 1 or 2, then $|d_{n}| \le 2 |d^{(1)}_{n-1}| \le 2^n$.
\item
If the first row is of type 3 or 4, then $|d_{n}| = | d^{(1)}_{n-1} \pm d^{(2)}_{n-1}| \le
| d^{(1)}_{n-1}| + |d^{(2)}_{n-1}|  \le 2^n$.
\item
If the first row is of type 5, then we permute the columns so that it becomes exactly as in (\ref{nine-types})
and then add the second column to the first one. Then, the first row becomes $(1,\,-1,\,0,\,\cdots,\,0)$.
The first minor does not change and contains rows such as $(x_2,\,x_3,\,\dots,\,x_n)$,
while the second minor contains rows $(x_1+x_2,\,x_3,\,\dots,\,x_n)$.
Due to the properties discussed above, these strings are also of type 1--9,
therefore the induction assumption applies to both minors.
We therefore conclude that
$|d_{n}| = | d^{(1)}_{n-1} + d^{(2)}_{n-1}| \le 2^n$.
\item
If the first row is of type 6, then repeat follow the same procedure with only change that we add the second column
to the third. The first row becomes $(1,\,1,\,0,\,\dots,\,0)$, while
the other rows have generic form $(x_1,\,x_2,\,x_2+x_3,\, x_4,\, \dots)$.
The first minor contains rows of the form $(x_2,\,x_2+x_3,\, x_4,\, \dots)$, which is
equivalent to $(x_2,\,x_3,\, x_4,\, \dots)$, while the second minor contains
$(x_1,\,x_2+x_3,\, x_4,\, \dots)$. Both rows are of type 1--9, therefore
$|d_{n}| = | d^{(1)}_{n-1} - d^{(2)}_{n-1}| \le 2^n$.
\item
If the first row is of type 7, we follow the procedure described for type 5 and get
$|d_{n}| = 2| d^{(1)}_{n-1}| \le 2^n$.
\item
If the first row is of type 8, we add the second and the third columns to the first, so that
the first row becomes $(0,\,-1,\,-1,\,\dots,\,0)$, while the other rows are of the form
$(x_1+x_2+x_3,\,x_2,\,x_3,\, x_4,\,\dots)$.
The second minor then contains rows $(x_1+x_2+x_3,\,x_3,\,x_4,\, \dots)$ which are equivalent
to $(x_1+x_2,\,x_3,\,x_4,\, \dots)$, being of allowed type.
The third minor contains $(x_1+x_2+x_3,\,x_2,\,x_4,\, \dots)$ equivalent
to $(x_1+x_3,\,x_2,\,x_4,\, \dots)$, again of allowed type.
Therefore the induction assumption applies to both minors, and we conclude that
$|d_{n}| = | d^{(2)}_{n-1} - d^{(3)}_{n-1}| \le 2^n$.
\item
If the first row is of type 9, we add the third column to the first and the fourth column to the second.
The first row turns into $(0,\,0,\,-1,\,-1,\,\dots,\,0)$, while the other rows become of form
 $(x_1+x_3,\,x_2+x_4,\,x_3,\,x_4,\, \dots)$.
The third minor is built of rows $(x_1+x_3,\,x_2+x_4,\,x_4,\, \dots)$, which are equivalent to
$(x_1+x_3,\,x_2,\,x_4,\, \dots)$, while the fourth minor is built of rows $(x_1+x_3,\,x_2+x_4,\,x_3,\, \dots)$ equivalent to
$(x_1,\,x_2+x_4,\,x_3,\, \dots)$. Both are of the allowed type, so the induction assumption
applies to both minors, and
$|d_{n}| = | d^{(4)}_{n-1} - d^{(3)}_{n-1}| \le 2^n$.
\end{enumerate}
Therefore, $|d_{n}| \le 2^n$ follows for any type of the first row, which completes the proof.
\end{proof}

\subsection{Cyclic groups and their products}

Here we prove two propositions which show that a rather broad class of finite abelian groups
are realizable as symmetry groups of the Higgs
potential in the NHDM. First we show which cyclic groups are realizable and then consider direct products
of cyclic groups.

We showed in the previous subsection that in order to find the order of the finite group,
one can use representation (\ref{XcA}) and focus on the matrix $c_{ij}$ instead of the matrix $X$.
Let us now prove that the matrix $c$ can also be used to find the finite symmetry group of the given potential.
First, recall that the matrix $X$ defines a system of linear equations on angles $\alpha_i$:
\be
X_{ij}\alpha_j = c_{ik}A_{kj} \alpha_j = 2\pi r_i\,.
\ee
This system remains the same upon
\begin{itemize}
\item
simultaneous sign change of the $k$-th column in $c$ and $k$-th row in $A$;
\item
simultaneous exchange of two columns in $c$ and two rows in $A$;
\item
simultaneous summation of two columns in $c$ and subtraction of two rows in $A$.
\item
sign changes of the $k$-th column in $A$ and $\alpha_k$;
\item
exchange of columns $i$ and $j$ in $A$ and exchange $\alpha_i \leftrightarrow \alpha_j$;
\item
summation of two columns in $A$ and subtraction of two $\alpha$'s;
\item
sign changes of the $k$-th row in $c$ and of the integer parameters $r_k$;
\item
exchange of two rows in $c$ and of two $r$'s;
\item
summation of two rows in $c$ and of two parameters $r$.
\end{itemize}
In short, it means that all allowed manipulations with integer matrices described earlier
can be used for $A$ and $c$.

Then we proceeding in the following way.
We first diagonalize $A$; the matrix $c$ becomes modified by a series of allowed transformations.
But diagonal $A$ is equal to the identity matrix. Therefore, we are left only with $c$ in the system of equations.
We then diagonalize $c$ and construct the symmetry group from its diagonal values.
This means that in order to prove that a given group is realizable, we simply need to give an example of matrix $c_{ij}$
constructed from rows of type 1--9, (\ref{nine-types}),
which yields the desired diagonal values after diagonalization.

Now we move to the two propositions.

\begin{proposition}\label{prop-cyclic-NHDM}
The cyclic group $\Z_p$ is realizable for any positive integer $p \le 2^n$.
\end{proposition}

\begin{proof}We start with the following $n \times n$ matrix $c_{ij}$:
\be
c_{2^n} =
\left(\begin{array}{ccccccc}
2 & -1 & 0 & 0 &\cdots & 0 & 0 \\
0 & 2 & -1 & 0 &\cdots & 0 & 0 \\
0 & 0 & 2 & -1 &\cdots & 0 & 0 \\
\vdots & \vdots & \vdots & \vdots & & \vdots & \vdots \\
0 & 0 & 0 & 0 &\cdots & 2 & -1 \\
0 & 0 & 0 & 0 & \cdots & 0 & 2
\end{array}\right)\,.\label{c2n}
\ee
By straightforward manipulation with columns, starting from the last one,
we arrive at
\be
c_{2^n} =
\left(\begin{array}{ccccccc}
0 & -1 & 0 & 0 &\cdots & 0 & 0 \\
0 & 0 & -1 & 0 &\cdots & 0 & 0 \\
0 & 0 & 0 & -1 &\cdots & 0 & 0 \\
\vdots & \vdots & \vdots & \vdots & & \vdots & \vdots \\
0 & 0 & 0 & 0 &\cdots & 0 & -1 \\
2^n & 2^{n-1} & 2^{n-2} & 2^{n-3} & \cdots & 2^2 & 2
\end{array}\right) \to
\left(\begin{array}{ccccccc}
2^n & 0 & 0 &\cdots & 0 & 0 \\
0 &  1 & 0 &\cdots & 0 & 0 \\
0 & 0 & 1 &\cdots & 0 & 0 \\
\vdots & \vdots & \vdots & & \vdots & \vdots \\
0 & 0 & 0 & \cdots & 0 & 1
\end{array}\right)\,,
\ee
which produces the finite group $\Z_{2^n}$.
Now, let us modify $c_{2^n}$ by replacing the zero at the left lower corner by $-1$.
Then the same transformation leads us to the group $\Z_{2^n-1}$.
If, instead, we replace any zero in the first column by $-1$, $c_{k1} = -1$,
then the diagonalization procedure leads us to the group $\Z_p$ with $p = 2^n - 2^{n-k}$.

Now, generally, take any integer $0 \le q < 2^{n}$ and write it in the binary form.
This form uses at most $n$ digits. Write this binary form as a vector with $n$ components
and subtract it from the first column of the matrix $c_{2^n}$. Then, after diagonalization,
we obtain the symmetry group $\Z_p$ with $p = 2^n - q$. This completes the proof.
\end{proof}

Just to illustrate this construction, for $n=5$ and $q = 23$ we write
$$
q = 23_{10} = 10111_2 = \left(\begin{array}{c}
1 \\
0 \\
1 \\
1 \\
1
\end{array}
\right); \quad c_{32} - q \equiv
\left(\begin{array}{ccccccc}
1 & -1 & 0 & 0 & 0 \\
0 & 2 & -1 & 0 & 0 \\
-1 & 0 & 2 & -1 & 0 \\
-1 & 0 & 0 & 2 & -1 \\
-1 & 0 & 0 & 0 & 2
\end{array}\right)\,.
$$
Repeating the above procedure we obtain at the lowest left corner $2^5 - 2^4 - 2^2 - 2^1 - 2^0 = 32 - 23 = 9$,
which gives the $\Z_9$ group.

Now we show that not only cyclic groups but many of their products can be realized as symmetry groups of some potential.

\begin{proposition}\label{prop-many-cyclic-NHDM}
Let $n = \sum_{i=1}^k n_i$ be a partitioning of $n$ into a sum of non-negative integers $n_i$.
Then, the finite group
\be
G = \Z_{p_1}\times \Z_{p_2}\times \cdots \times \Z_{p_k}\label{manyZn}
\ee
is realizable for any $0 < p_i < 2^{n_i}$.
\end{proposition}

\begin{proof}
Let us start again from the matrix (\ref{c2n}). For any partitioning of $n$ one can turn this matrix into a block-diagonal
matrix by replacing some of the $-1$'s by zeros. The matrix is then represented by the smaller square blocks of size
$n_1,\, n_2,\, \dots ,\, n_k$. The $i$-th blocks has exactly the form of (\ref{c2n}), with $n$ replaced by $n_i$.
Therefore, within this block one can encode any cyclic group $\Z_{p_i}$, with $0 < p_i \le 2^{n_i}$.
Since each block can be treated independently, we can encode any group of the form (\ref{manyZn}).
\end{proof}

Let us note that this proposition covers many but not all possible finite abelian groups with order $\le 2^n$.
For example, for $n=5$, we can think of the group $\Z_5 \times \Z_5$ whose order is smaller than $2^5 = 32$.
However, there exists no partitioning of $5$ that would lead to this group by applying the proposition just proved.
At this moment, it remains an open problem if such groups are realizable.
Using Theorem~\ref{theorem-order}, we can formulate the following conjecture:
\begin{conjecture}\label{conjecture-list}
Any finite abelian group with order $\le 2^{N-1}$ is realizable in NHDM.
\end{conjecture}
If proven true, this conjecture will give the complete classification of realizable subgroups of the maximal torus in NHDM.

\section{Conclusions}\label{section-discussion}

Models with several Higgs doublets have become a popular framework
to introduce physics beyond the Standard Model.
When building the scalar sector of an $N$-doublet model,
one usually imposes various discrete symmetries, and many specific symmetry groups
have been used.
However, the complete lists of groups that can be implemented for a given $N$ are unknown for $N>2$.

In this paper we made a step towards classification of possible symmetries of the scalar sector of the NHDM.
Namely, we studied which abelian groups can be realized as symmetry groups of the NHDM potential.
We proved that they can be either subgroups of the maximal torus,
or the image under the canonical homomorphism $SU(N) \to PSU(N)$ of a finite nilpotent group of class 2.
For the subgroups of the maximal torus,
we developed an algorithmic strategy that gives full list of possible realizable abelian symmetries for any given $N$.
We illustrated how the strategy works with two small-$N$ examples. For 3HDM we gave the full list of realizable abelian symmetries,
including the ones with generalized $CP$-transformations, while for the 4HDM we just listed all finite realizable subgroups of the maximal torus.
We also proved that the order of any realizable finite group in NHDM is $\le 2^{N-1}$ and explicitly described
which finite groups can appear at given $N$.
Finally, we conjectured that {\em any} finite abelian group with order $\le 2^{N-1}$ is realizable in the NHDM.

It is worth stressing that the Higgs potential taken in the usual quadratic plus quartic form plays a major factor in restricting
possible symmetry groups from {\em all} subgroups of $SU(N)$.
If the potential could contain higher-order terms, the list of realizable group would be larger.
It means, in particular, that the recent impressive works on the beyond-SM applications of the finite subgroups of $SU(3)$,
\cite{Grimus}, cannot be directly applied to the classification of the Higgs symmetries in 3HDM because
no criterion was given which of the groups are realizable and which are not.

An obvious direction of future research is to understand phenomenological consequences of the symmetries found.
This includes, in particular, a study of how the symmetries of the potential
are broken, and what is the effect of these symmetries on the physical Higgs boson properties.
One particular result already obtained here is the possibility to accommodate scalar dark matter stabilized by the $Z_p$
symmetry group, for any $p \le 2^{N-1}$, \cite{ZpDM}.

\section*{Acknowledgements}
We are grateful to Pedro Ferreira and Jo\~ao Silva for valuable comments and suggestions.
This work was supported by the Belgian Fund F.R.S.-FNRS via the
contract of Charg\'e de recherches, and in part by grants RFBR No.11-02-00242-a and No.10-01-0391, and NSh-3810.2010.2.

\appendix

\section{Maximal abelian subgroups in $PSU(N)$}\label{appendix-maximal-abelian-PSU}

Let $V=\mathbb{C}^N$ be a unitary space of dimension $N$ equipped with unitary
form $(u,v)$. If $L$ is a subspace of $V$, then $L^\dagger$ denotes the
orthogonal complements to $L$ in $V$. Now $SU(N)$ is the group of all
transformations of determinant $1$ preserving the form, i.e.,
$$SU(N)=\{\x\mid \det \x=1 \text{ and } \x^\dagger\cdot \x=1\}.$$
We denote the group of all transformations preserving the form by $GU(N)$.
Clearly $SU(N)$ is a normal subgroup of $GU(N)$ ($SU(N)\unlhd GU(N)$).

In the following lemma we collect known facts and we provide the proofs for the
reader's convenience.

\begin{lemma}\label{SimpleFactsUnitary}
Let $L$ be a subspace of $V$ and $\x$ be a linear transformation of $V$.
Then the following statements hold.
\begin{itemize}
 \item[{\em (1)}] If $L$ is $\x$-invariant, then $L^\dagger$ is
$\x^\dagger$-invariant.
\item[{\em (2)}] If $\x\in GU(N)$, then $\x(v)=\lambda v$ implies
$\x^\dagger(v)=\bar\lambda v$.
\item[{\em (3)}] If $\x\in GU(N)$, then there exists an orthonormal basis
consisting of eigenvectors of~$\x$. Moreover, if
$Spec(\x)=\{\lambda_1,\ldots,\lambda_k\}$ is the set of all eigenvalues of
$\x$ and  $V_\lambda=\{v\in V\mid \x(v)=\lambda v\}$, then there
exists the orthogonal decomposition
$V=V_{\lambda_1}\oplus^\dagger\ldots\oplus^\dagger V_{\lambda_k}$.
\item[{\em (4)}] If $\x,\y\in GU(N)$ commute, $\lambda$ is an eigenvalue
of $\x$, then $V_\lambda$ is $\y$-invariant.
\item[{\em (5)}] If $A\leq GU(N)$ is abelian, then $A$ is conjugate to a
subgroup of $D(N)$in $GU(N)$, where $D(N)$ is a group of all diagonal matrices
in $GU(N)$. In particular, if $A\leq SU(N)$ is abelian, then $A$ is conjugate
to a subgroup of $D(N)\cap SU(N)$ in~$SU(N)$.
\item[{\em (6)}] Suppose $H$ is the stabilizer in $SU(N)$ of $L$ and
$L^\dagger$, i.e., $H=( GU(L)\times GU(L^\dagger))\cap SU(N)$. Assume that
$L,L^\dagger\not=0$. Let
$\overline{\phantom{G}}:SU(N)\rightarrow PSU(N)$ be the reduction modulo
scalars. Then the natural
projection $\pi:\overline{H}\rightarrow GU(L)$ is surjective. Moreover, if $T$
is a maximal torus of $H$, then the projection of $T$ into $GU(L)$ is
a maximal torus of $GU(L)$.
\end{itemize}
\end{lemma}

\begin{proof}
By definition of $\x^\dagger$, for every $u,v\in V$ we have
$(\x(u), v)=(u,\x^\dagger(v))$, whence~(1).

(2) Assume that $\x(v)=\lambda v$. Then
\begin{multline*}
0=(\x(v)-\lambda v,\x(v)-\lambda
v)=((\x-\lambda1)(v),(\x-\lambda1)(v))
=(v,(\x^\dagger-\bar\lambda1)(\x-\lambda1)(v))=\\(v,
(\x-\lambda1)(\x^\dagger-\bar\lambda1)(v))=(
(\x^\dagger-\bar\lambda 1)(v),(\x^\dagger-\bar\lambda1)(v)),
\end{multline*}
whence $\x^\dagger(v)-\bar\lambda v=0.$

(3) Since $\mathbb{C}$ is algebraically closed, there exists an eigenvalue $\lambda$ of $\x$. Therefore there exists an eigenvector $v$ such that $\x(v)=\lambda v$. Clearly we may assume that $\|v\|=1$. Consider the line $L(v)$ spanned by $v$. In view of item (2) of the lemma, $L(v)$ is $\x^\dagger$-invariant, hence item (1) of the lemma implies that $L(v)^\dagger$ is $\x$-invariant. Therefore we obatin a decomposition $V=L(v)\oplus L(v)^\dagger$ of $V$ into an orthogonal sum of $\x$-invariant subspaces. Induction on $\dim(V)$ completes the proof of this item.

(4) For every $v\in V_\lambda$ we have $\x(\y(v))=\y(\x(v))=\lambda \y(v)$, i.e., $\y(v)\in V_\lambda$. Therefore $V_\lambda$ is a $\y$-invariant subspace.

(5) Assume that $A\leq GU(N)$ is abelian. Clearly we may assume that $A$ is
nonscalar. Choose $x\in A\setminus Z(GU(N))$. Given $\lambda \in \mathbb{C}$ consider $V_\lambda=\{v\in V\mid
\x(v)=\lambda v\}$. Denote by $Spec(\x)$ the set of all eigenvalues of $\x$. Since $\mathbb{C}$ is algebraically closed, we obtain that $Spec(\x)$ is nonempty. Moreover, since $\x\in A\setminus Z(GU(N))$, it follows that $Spec(\x)$ contains at least two elements. In view of the proof of item (3) we obtain the orthogonal decomposition $$V=\oplus_{\lambda\in Spec(\x)} V_\lambda.$$ Choose $y\in A$. Then, for every $\lambda\in Spec(\x)$ the subspace $V_\lambda$ is $y$-invariant, hence $V_\lambda$ is $A$-invariant.  Induction on $\dim(V)$ completes the proof of this item.

(6) A matrix  from $( GU(L)\times GU(L^\dagger))\cap SU(N)$ has the following
form: $$\left(\begin{array}{cc}
               A&0\\
               0&B
              \end{array}\right)$$ and $\vert A\vert \cdot \vert B\vert=1$.
Assume that $A\in GU(L)$ and $\vert A\vert=\mu$. Then $\vert \mu\vert=1$, i.e.
$\mu\cdot \mu^\ast=1$, so if $D=diag(\mu^\ast,1,\ldots,1)$, then
$$\x=\left(\begin{array}{cc}
               A&0\\
               0&D
              \end{array}\right)\in H.$$ Moreover $\pi(\x)=A$, whence~(6).

\end{proof}

Consider the canonical epimorphism
\be
\overline{\phantom{G}}:\ SU(N)\to SU(N)/Z(SU(N))=\overline{SU(N)} \simeq PSU(N)\,,
\ee
where $Z(SU(N)) \simeq \Z_N$ is the center of $SU(N)$. Given subgroup $\overline{A}$ of $PSU(N)$ we denote by $A$ the complete preimage of $\overline{A}$ in $SU(N)$. The following statement is evident.

\begin{proposition}\label{prop-abelian-PSU}
Let $\overline{A}$ be an abelian subgroup of $PSU(N)$.
Then $A$ is either an abelian group or a nilpotent group of class $2$ and $Z(SU(N))\leq Z(A)$.
\end{proposition}

\begin{proof}
Since $Z(SU(N))\leq Z(A)$ for every $A\leq SU(N)$, the statement is evident.
\end{proof}

If $\overline{A}$ is a maximal abelian subgroup of $PSU(N)$ such that $A$ is abelian, then $A$ must be a maximal abelian subgroup of $SU(N)$,
and item (5) of Lemma \ref{SimpleFactsUnitary} implies that $A$ is a maximal torus.

Now we find the structure of $A$ in case, when $A$ is nonabelian. We assume
that $\overline{A}$ is a maximal abelian subgroup of~$PSU(N)$.

\begin{lemma}\label{CenterOfA}
If $A$ is nonabelian, then $Z(A)=Z(SU(N))$.
\end{lemma}

\begin{proof}
 Inclusion $Z(SU(N))\leq Z(A)$ is evident. Assume that there exists $\x\in
Z(A)\setminus Z(SU(N))$. Then $\x$ is nonscalar, so
$Spec(\x)=\{\lambda_1,\ldots,\lambda_k\}$ and $k\ge2$. In view of Lemma
\ref{SimpleFactsUnitary} items (3) and (4), $A\leq
C_{SU(N)}(\x)=GU(N_1)\times\ldots\times GU(N_k)$, where
$N_i=\dim(V_{\lambda_i})$. Now Lemma \ref{SimpleFactsUnitary} items (5) and (6)
implies that the projection of $\overline{A}$ into $GU(N_i)$ is contained in a
maximal torus of $GU(N_i)$. Since $\overline{A}$ is a maximal abelian subgroup
of $PSU(N)$, we obtain that the projection of $A$ coincides with a maximal
torus of  $GU(N_i)$ for every $i$. Hence $\overline{A}$ is   a
maximal torus of $PSU(N)$, and so $A$ is    a maximal torus of
$SU(N)$. In particular, $A$ is abelian, a contradiction.
\end{proof}

Since $Z(A)=Z(SU(N))\simeq \mathbb{Z}_N$, we obtain that the exponent of $Z(A)$
equals $N$. In view of \cite[5.2.22(i)]{Rob}, the exponent of $A$ divides
$N^2$, in particular, $\pi(A)= \pi(N)$,
where $\pi(A)$ denotes the set of prime divisors of $\vert A\vert$ and $\pi(N)$ denotes the set of prime divisors of $N$.

\begin{lemma}\label{A_1IsFinite}
$A$ is finite.
\end{lemma}

\begin{proof}
Let $H$ be a maximal normal abelian subgroup of $A$. Then $C_{A}(H)=H$ (see
\cite[5.2.3]{Rob}) and $H$ is contained in a maximal torus of
$SU(N)$. In particular, $H$ is generated by a finite number of elements. Since
$\pi(H)\subseteq \pi(A)=\pi(N)$ is finite and the exponent of $H$ divides $N^2$, we obtain that $H$ is finite. Since
$H=C_{A}(H)$, it follows that $A/H$ is isomorphic to a subgroup of $Aut(H)$,
in particular $A/H$ is finite. Therefore $A$ is finite.
\end{proof}

All the information about abelian subgroups of $PSU(N)$ obtained above
was assembled in Theorem~\ref{AbelianSubgroupsPSUN} in the main text.

We finish by giving an example of subgroups satisfying statement (2) of Theorem~\ref{AbelianSubgroupsPSUN}.
Let $\omega$ be the primitive $N$-th root of $1$ in $\mathbb{C}$. Consider
$\x=diag(1,\omega,\omega^2,\ldots,\omega^{N-1})$,
$\y=e_{1,2}+e_{2,3}+\ldots+e_{N-1,N}+e_{N,1}$, where $e_{i,j}$ is the matrix
with $1$ on the $(i,j)$-place and $0$ on the remaining places. Then $M=\langle
\x,\y\rangle$ is a nilpotent subgroup of class $2$ of $SU(N)$ for $N$ odd,
$Z(N)=Z(SU(N))$ and the exponent of $M$ equals~$N$.

\section{Abelian symmetries in 3HDM with antiunitary transformations}\label{appendix-antiunitary}

Here we explicitly describe embedding of each of the unitary abelian groups (\ref{list3HDM}) into a larger abelian group
that contains antiunitary transformations.
Throughout this section we will use the following notation: if a cyclic group $\Z_q$ is generated by an antiunitary
transformation, we indicate it by an asterisk: $\Z_q^*$.

\subsection{Embedding $U(1)_1\times U(1)_2$}
We proved in the general case that the maximal torus cannot be embedded in a larger abelian group.
For illustration, let us repeat the argument in the specific case of 3HDM.
At the first step of the strategy we search for a matrix $b\in PSU(3)$ that satisfies $aba=b$ for all $a \in T$.
Since $a$ is diagonal, each element $(aba)_{ij}$ is equal to $b_{ij}$ up to a phase rotation $\psi_{ij}$:
\be
\psi_{ij} =  \left(\begin{array}{ccc}
-2\alpha - 4\beta/3 & -\beta/3 & -\alpha - \beta/3 \\
- \beta/3 & 2\alpha + 2\beta/3 & \alpha + 2\beta/3 \\
-\alpha - \beta/3 & \alpha + 2\beta/3 & 2\beta/3
\end{array}\right)\,.\label{psi-ij}
\ee
In order for an element $b_{ij}$ to be non-zero, the corresponding phase rotation must be zero or multiple of $2\pi$,
modulo to $2\pi k/3$ along the diagonal.
It is impossible to find such $b$ for arbitraty $\alpha$ and $\beta$. Therefore. $T = U(1)_1 \times U(1)_2$
cannot be embedded into a larger abelian group.

\subsection{Embedding $U(1)$}

There are two distinct sorts of $U(1)$ groups inside $T$: $U(1)_1$-type and $U(1)_2$-type.
Consider first $U(1)_1$.
The phase rotations for the group $U(1)_1$ ($\beta=0$, arbitrary $\alpha$) allow $b$ to have non-zero entries $b_{12}$, $b_{21}$ and $b_{33}$
only. We are free to choose $b_{12} = b_{21} = 1$, $b_{33}=1$ (the fact that $\det b = -1$ instead of $1$ is inessential;
we can always multiply all the transformations by the overall $-1$ factor). Then, the transformation $J'$ which commutes with any $a \in U(1)_1$ is
\be
J' =  \left(\begin{array}{ccc}
0 & 1 & 0 \\
1 & 0 & 0 \\
0 & 0 & 1
\end{array}\right)\cdot J\,,\quad
aJ' = J'a \quad \forall a \in U(1)_1\,.
\label{J'-in-U1}
\ee
Next, we search for transformations $x \in PSU(3)$ such that $xJ'$ also commutes with any element $a \in U(1)_1$;
such transformations $x$ form the centralizer of $U(1)_1$.
From the equation $axa^{-1} = x$, analyzed with the same technique of phase rotations,
we find that $x$ must be diagonal. By checking that the product of $x_1J'$ and $x_2J'$ stays inside $U(1)_1$, we can conclude
that $x$ must belong to one of the $U(1)_1$-type groups $X_{\xi}$:
\be
x =  \left(\begin{array}{ccc}
e^{-i\gamma} & 0 & 0 \\
0 & e^{i\gamma} & 0 \\
0 & 0 & 1
\end{array}\right)
\left(\begin{array}{ccc}
e^{i\xi} & 0 & 0 \\
0 & e^{i\xi} & 0 \\
0 & 0 & e^{-2i\xi}
\end{array}\right)\,. \label{xgamxi}
\ee
Here $\xi \in [0,2\pi/3)$ is an arbitrary but fixed parameter specifying which $X_{\xi}$ group we take,
while $\gamma$ is the running angle parametrizing the elements of this group.
At this stage, any choice of $\xi$ is acceptable.
Thus, we embedded $U(1)_1$ into an abelian group generated by
$\langle U(1)_1, J''_{\xi}\rangle \simeq U(1) \times \Z_2^*$,
where
\be
J''_{\xi} =
\left(\begin{array}{ccc}
e^{i\xi} & 0 & 0 \\
0 & e^{i\xi} & 0 \\
0 & 0 & e^{-2i\xi}
\end{array}\right)\cdot J' =
\left(\begin{array}{ccc}
0 & e^{i\xi} & 0 \\
e^{i\xi} & 0 & 0 \\
0 & 0 & e^{-2i\xi}
\end{array}\right)\cdot J\,.\label{J''xi}
\ee

Let us now see whether this group is realizable and which $\xi$ must be chosen.
The only $U(1)_1$ symmetric terms in the potential are
\be
\lambda (\phi_1^\dagger \phi_3)(\phi_2^\dagger \phi_3)  + \lambda^* (\phi_3^\dagger \phi_1)(\phi_3^\dagger \phi_2)\,.
\label{extratermsU1}
\ee
This sum of two terms is indeed invariant under $J''_{\xi}$ provided $\xi = \psi_\lambda/3$, where $\psi_\lambda$
is the phase of $\lambda$. This prescription uniquely specifies which group $X_\xi \subset PSU(3)$.
Besides, in order to guarantee that the $T$-symmetric terms (\ref{3HDMpotential-Tsymmetric}) are invariant
under $J''$, we must set
\be
m_{11}^2=m_{22}^2\,,\quad \lambda_{11}=\lambda_{22}\,,\quad \lambda_{13} = \lambda_{23}\,,\quad \lambda'_{13} = \lambda'_{23}\,.
\label{U1-restrictions}
\ee
However upon this reduction of free parameters we observe that the resulting potential
acquires an additional {\em unitary} symmetry: the exchange of the field two doublets $\phi_1 \leftrightarrow \phi_2$.
This transformation does not commute with $U(1)_1$. Therefore, the true symmetry group of such potential is non-abelian, though it contains
the desired abelian subgroup $U(1) \times \Z_2^*$.
According to our definition, we conclude that $U(1) \times \Z_2^*$ is not realizable in 3HDM.

Consider now $U(1)_2$: $\alpha=0$, arbitrary $\beta$. It follows from (\ref{psi-ij}) that no non-trivial matrix $b \in PSU(3)$
can satisfy $aba = b$. Therefore,  $U(1)_2$ cannot be embedded into a larger abelian group with antiunitary transformations.
Note also that selecting various diagonal $U(1)$ subgroups of $T$, for example by setting $\alpha+\beta/3 = 0$, will result
just in another version of $U(1)_1$.
Finally, we note that the unitary $U(1)\times \Z_2$ emerges only when the continuous group is of the $U(1)_2$-type.
Therefore,  $U(1)\times \Z_2$ cannot be embedded into a larger realizable abelian group as well.

The overall result of the last two subsections is that in 3HDM all realizable continuous abelian groups are necessarity unitary and cannot contain
antiunitary transformations.

\subsection{Embedding $\Z_4$}

The $\Z_4$ symmetry group with generator $a$ arises as a subgroup of $U(1)_1$. Therefore, $b$ and $J'$ can be chosen
as before, (\ref{J'-in-U1}), but the new conditions for $x$ should now be checked.
This transformation can still be represented as in (\ref{xgamxi}), but the angle $\gamma$ can take discrete values:
\be
\gamma = {\pi\over 2}k\quad \mbox{or} \quad \gamma = {\pi \over 4} + {\pi\over 2}k\,,\quad k \in \Z\,.
\ee
Correspondingly, two possibilities for embedding the $\Z_4$ group arise: $\Z_4 \times \Z_2^*$ and $\Z_8^*$.

Consider first the $\Z_4 \times \Z_2^*$ embedding. The $\Z_2^*$ subgroup is generated $J''$ in (\ref{J''xi}),
where, as before, $\xi =  \psi_\lambda/3$.
As before, the potential satisfies (\ref{U1-restrictions}), includes (\ref{extratermsU1}), but in addition it now contains
\be
\lambda' (\phi_1^\dagger\phi_2)^2 + \lambda^{\prime *} (\phi_2^\dagger\phi_1)^2\,,\label{extratermZ4}
\ee
Note that these terms are invariant under $J''$, and therefore, under the full $\Z_4 \times \Z_2^*$.
which is also $J''$-symmetric when $\xi_1 = \xi_2$.
However, just as in the previous case, restrictions (\ref{U1-restrictions}) result in a potential which is symmetric under another $\Z_2$ group
generated by the unitary transformation
\be
d =  \left(\begin{array}{ccc}
0 & e^{i\delta} & 0 \\
e^{-i\delta} & 0 & 0 \\
0 & 0 & 1
\end{array}\right)\quad \mbox{with}\quad \delta = {\psi_{\lambda'} \over 2}\,, \label{dZ2}
\ee
where $\psi_{\lambda'}$ is the phase of $\lambda'$ in (\ref{extratermZ4}). Since $d$ does not commute with
phase rotations, therefore, the resulting symmetry group of the potential is non-abelian.
Hence, the  $\Z_4 \times \Z_2^*$ symmetry group is not realizable.

Consider now the $\Z_8^*$ group generated by
\be
J'''_{\xi} =  \left(\begin{array}{ccc}
0 & e^{i\xi - i\pi/4} & 0 \\
e^{i\xi+i\pi/4} & 0 & 0 \\
0 & 0 &  e^{-2i\xi}
\end{array}\right)\cdot J\label{J'''}
\ee
with the property $(J''')^2 = a$, the generator of the $\Z_4$.
Note that upon action of $J'''$ the term $(\phi_1^\dagger \phi_2)^2$ just changes its sign; therefore,
such term cannot appear in the potential. But in this case we are left with only one type of $\Z_4$-symmetric
terms. Then, according to the general discussion, rank$X(V) = 1$ and, therefore, the potential
becomes symmetric under a continuous symmetry group.
Thus, $\Z_8^*$ is not realizable as a symmetry group of a 3HDM potential.

\subsection{Embedding $\Z_3$}

The $\Z_3$ subgroup of the maximal torus $T$ is generated by the transformation $a$ which rotates the
phases of the doublets by $(-2\pi/3,\, 2\pi/3,\,0)$.
This group also arises from the $U(1)_1$-type group, therefore the representation for $J'$ in (\ref{J'-in-U1})
and $x$ in (\ref{xgamxi}) are still valid.
Closing the group under product requires that $\gamma$ in (\ref{xgamxi}) is a multiple of $\pi/3$.
We can introduce
\be
J''_{\xi} =  \left(\begin{array}{ccc}
0 & e^{i\xi-i\pi/3} & 0 \\
e^{i\xi+i\pi/3} & 0 & 0 \\
0 & 0 & e^{-2i\xi}
\end{array}\right)\cdot J\,,\label{J'-Z3}
\ee
and since $(J'')^2 = a$, this transformation generates the group $\Z_6^*$.

The $\Z_3$-symmetric terms in the potential are
\be
\lambda_1 (\phi_1^\dagger \phi_2)(\phi_1^\dagger \phi_3) +
\lambda_2 (\phi_2^\dagger \phi_3)(\phi_2^\dagger \phi_1) +
\lambda_3 (\phi_3^\dagger \phi_1)(\phi_3^\dagger \phi_2) + h.c.\label{Z3symmetric}
\ee
with complex $\lambda_i$.
These terms become symmetric under $\Z_6^*$, if the following conditions are satisfied:
\be
|\lambda_1|  = |\lambda_2| \quad \mbox{and} \quad \psi_1 + \psi_2 + \psi_3 = \pi\,,\label{Z6conditions}
\ee
where $\psi_i$ are phases of $\lambda_i$.
Indeed, one can check that in this case the sum (\ref{Z3symmetric}) is left invariant under $J''_\xi$
with $\xi = -\psi_3/3$.
It is, therefore, possible to have a $\Z_6^*$-symmetric potential by requiring (\ref{Z6conditions}) and applying the usual
restrictions (\ref{U1-restrictions}).

However, just as in the previous cases, the resulting potential becomes symmetric under a larger unitary symmetry group,
and therefore $\Z_6^*$ is not realizable as well. The transformation that leaves the potential symmetric
is of type (\ref{dZ2}) with $\delta = (\psi_1-\psi_2)/3$.

\subsection{Embedding $\Z_2\times \Z_2$}

The $\Z_2 \times \Z_2 \subset T$ symmetry can be realized as a group of simultaneous sign flips of any pair of doublets
(or alternatively of independent sign flips of any of the three doublets).
If the simultaneous sign flip of doublets $i$ and $j$ are denoted as $R_{ij}$,
then the group is ${1,\,R_{12},\, R_{13},\, R_{23}}$, with $R_{12}R_{13}=R_{23}$ and any $(R_{ij})^2=1$.
This group can be embedded into $\Z_2 \times \Z_2 \times \Z_2^*$, where $\Z_2^*$ is generated by
\be
J'_{\xi_1,\xi_2} =  \left(\begin{array}{ccc}
e^{i\xi_1} & 0 &0 \\
0 & e^{i\xi_2} & 0 \\
0 & 0 & e^{-i\xi_1 - i\xi_2}
\end{array}\right)\cdot J\,.\label{J'-Z2Z2}
\ee
At this moment, any pair $\xi_1$, $\xi_2$ can be used to construct the $\Z_2 \times \Z_2 \times \Z_2^*$ group.
Note also that this group can be also written as $\Z_2 \times \Z_2^* \times \Z_2^*$ or $\Z^*_2 \times \Z^*_2 \times \Z_2^*$ by redefining the generators.

The $\Z_2 \times \Z_2$-symmetric potential contains the following terms
\be
\tilde \lambda_{12} (\phi_1^\dagger \phi_2)^2 +
\tilde \lambda_{23} (\phi_2^\dagger \phi_3)^2 +
\tilde \lambda_{31} (\phi_3^\dagger \phi_1)^2 + h.c.\label{termsZ2Z2}
\ee
with complex $\tilde\lambda_{ij}$. If there phases are denoted as $\psi_{ij}$, then the condition
when (\ref{termsZ2Z2}) are symmetric under some $\Z_2^*$ with an appropriate choice of $\xi_1$, $\xi_2$ is
\be
\psi_{12} + \psi_{23} + \psi_{31} = 0\,,\label{Z2Z2Z2condition}
\ee
or alternatively, when the product $\tilde\lambda_{12} \tilde\lambda_{23} \tilde\lambda_{31}$ is real.

Note that in contrast to the previous cases, the full $T$-symmetric potential (\ref{3HDMpotential-Tsymmetric})
is invariant under $\Z_2 \times \Z_2 \times \Z_2^*$. This guarantees that no other unitary symmetry arises
in this case, and therefore this symmetry group is realizable.

\subsection{Embedding $\Z_2$}

The unitary abelian symmetry group $\Z_2$ can be implemented in a variety of ways, all of which are equivalent.
One can take, for example, the $\Z_2$ group generated by sign flips of the first and second doublets, $R_{12}$.
Again, the usual $CP$-transformation $J$ commutes with $R_{12}$.
An analysis similar to the previous case shows that the $\Z_2 \times \Z_2^*$ group generated by $R_{12}$ and $J$
is realizable.

There is also a possibility to embed the $\Z_2$ group into $\Z_4^*$. Consider the following
transformation
\be
J'_{\xi} =  \left(\begin{array}{ccc}
0 & e^{i\xi} & 0 \\
-e^{i\xi} & 0 & 0 \\
0 & 0 & e^{-2i\xi}
\end{array}\right)\cdot J\,,\quad (J')^2 = R_{12}\,,\label{J'-Z2}
\ee
where $\xi$ is an arbitrary but fixed parameter.
The potential with the usual restrictions (\ref{U1-restrictions}) and with extra terms
\be
\lambda_5(\phi_1^\dagger \phi_2)^2 + \lambda_6(\phi_1^\dagger\phi_2)\left[(\phi_1^\dagger\phi_1)-(\phi_2^\dagger\phi_2)\right]
+ \lambda_7 (\phi_1^\dagger\phi_3)(\phi_2^\dagger\phi_3)
+\lambda_8 (\phi_1^\dagger\phi_3)^2 + \lambda_9 (\phi_2^\dagger\phi_3)^2 + h.c.
\ee
satisfying conditions $|\lambda_8|=|\lambda_9|$ and $\psi_8 + \psi_9 = 2\psi_7 + \pi$ is invariant under $J'_{\xi}$
with $6\xi = \psi_8 + \psi_9$.
Adjusting $\lambda_5$ and $\lambda_6$, one can guarantee that no other unitary symmetry appears as a result of (\ref{U1-restrictions}).
Therefore, this $\Z_4^*$ group is realizable.


\begin{thebibliography}{99}
\bibitem{CPNSh}
E.~Accomando {\it et al.}, ``Workshop on CP studies and non-standard Higgs physics,''
  arXiv:hep-ph/0608079.
\bibitem{Haber}
J.F.~Gunion, H.E.~Haber, G.~Kane, S.~Dawson, {\em The Higgs
Hunter's Guide} (Addison-Wesley, Reading, 1990).
\bibitem{ZpDM} I.~P.~Ivanov, V.~Keus,
arXiv:1203.3426 [hep-ph].
\bibitem{Lee}
T.~D.~Lee, Phys. Rev. D {\bf 8}, 1226 (1973).
\bibitem{review2011}
G.~C.~Branco et al,
arXiv:1106.0034 [hep-ph].
\bibitem{CPexamples}
M.~Maniatis, A.~von Manteuffel and O.~Nachtmann,
  Eur.\ Phys.\ J.\ C {\bf 57}, 739 (2008);
M.~Maniatis and O.~Nachtmann,
  JHEP {\bf 0905}, 028 (2009);
P.~M.~Ferreira and J.~P.~Silva,
  Eur.\ Phys.\ J.\ C {\bf 69}, 45 (2010).
\bibitem{frobenius}
S.~F.~King, C.~Luhn,
  JHEP {\bf 0910}, 093 (2009);
Q.~-H.~Cao, S.~Khalil, E.~Ma, H.~Okada,
  Phys.\ Rev.\ Lett.\  {\bf 106}, 131801 (2011).
\bibitem{Botella}
F.~J.~Botella and J.~P.~Silva,
  Phys.\ Rev.\  D {\bf 51}, 3870 (1995);
G.~C.~Branco, L.~Lavoura and J~.P.~Silva, ``CP-violation'',
Oxford University Press, Oxford, England (1999).
\bibitem{Davidson}
S.~Davidson and H.~E.~Haber,  Phys.\ Rev.\ D {\bf 72}, 035004 (2005)
  [Erratum-ibid.\ D {\bf 72}, 099902 (2005)];
  H.~E.~Haber and D.~O'Neil, Phys.\ Rev.\ D {\bf 74}, 015018 (2006)
  [Erratum-ibid.\ D {\bf 74}, 059905 (2006)].
\bibitem{Gunion}
J.~F.~Gunion and H.~E.~Haber,  Phys.\ Rev.\ D {\bf 72}, 095002 (2005).
\bibitem{Oneil}
 D.~O'Neil,
  ``Phenomenology of the Basis-Independent CP-Violating Two-Higgs Doublet Model'',
  Ph.D. thesis, University of California Santa Cruz (2009),
  arXiv:0908.1363 [hep-ph].
\bibitem{Sartori}
 G.~Sartori and G.~Valente, arXiv:hep-ph/0304026.
\bibitem{Nagel}
F. Nagel, ``New aspects of gauge-boson couplings and the Higgs sector'', Ph.D. thesis, University Heidelberg
(2004), [http://www.ub.uni-heidelberg.de/archiv/4803].
\bibitem{Maniatis}
 M.~Maniatis, A.~von Manteuffel, O.~Nachtmann and F.~Nagel,  Eur.\ Phys.\ J.\  C {\bf 48}, 805 (2006);
M.~Maniatis, A.~von Manteuffel and O.~Nachtmann, Eur.\ Phys.\ J.\  C {\bf 57}, 719 (2008).
\bibitem{Nishi0}
C.~C.~Nishi, Phys.\ Rev.\ D {\bf 74}, 036003 (2006)
[Erratum-ibid.D {\bf 76}, 119901 (2007)].
\bibitem{Ivanov0}
 I.~P.~Ivanov,  Phys.\ Lett.\ B {\bf 632}, 360 (2006);
Phys.\ Rev.\ D {\bf 75}, 035001 (2007) [Erratum-ibid.  D {\bf 76}, 039902 (2007)];
Phys.\ Rev.\ D {\bf 77}, 015017 (2008).
\bibitem{generalizedCP}
P.~M.~Ferreira, H.~E.~Haber, M.~Maniatis, O.~Nachtmann, J.~P.~Silva,
  Int.\ J.\ Mod.\ Phys.\  {\bf A26}, 769-808 (2011).
  [arXiv:1010.0935 [hep-ph]].
\bibitem{pilaftsis}
A.~Pilaftsis,
  Phys.\ Lett.\ B {\bf 706}, 465 (2012).
\bibitem{Erdem}R.~Erdem,
  Phys.\ Lett.\  B {\bf 355}, 222 (1995).
\bibitem{Barroso}A.~Barroso, P.~M.~Ferreira, R.~Santos and J.~P.~Silva,
  Phys.\ Rev.\  D {\bf 74}, 085016 (2006).
\bibitem{Nishi1} C.~C.~Nishi,
  Phys.\ Rev.\  D {\bf 76}, 055013 (2007).
\bibitem{Ferreira}P.~M.~Ferreira and J.~P.~Silva,
  Phys.\ Rev.\  D {\bf 78}, 116007 (2008).
\bibitem{Grimus} W.~Grimus and P.~O.~Ludl,
J.\ Phys.\ A {\bf 43}, 445209 (2010); \\
W.~Grimus and P.~O.~Ludl, arXiv:1110.6376 [hep-ph].
\bibitem{Olaussen} K.~Olaussen, P.~Osland and M.~A.~Solberg,
  JHEP {\bf 1107}, 020 (2011).
\bibitem{Ivanov1}I.~P.~Ivanov, C.~C.~Nishi, Phys.\ Rev.\ D {\bf 82}, 015014 (2010).
\bibitem{Ivanov2}I.~P.~Ivanov, JHEP {\bf 1007}, 020 (2010).
\bibitem{Machado}A.~C.~B.~Machado, J.~C.~Montero, V.~Pleitez, Phys.\ Lett.\ B {\bf 697}, 318-322 (2011).
\bibitem{Fukuyama}T.~Fukuyama, H.~Sugiyama, K.~Tsumura, Phys.\ Rev.\ D {\bf 82}, 036004 (2010).
\bibitem{Adler}S.~Adler, Phys.\ Rev.\ D {\bf 60}, 015002 (1999).
\bibitem{Pontr} L.~S.~Pontryagin, Ann. of Math. {\bf 35}, 361-388 (1934). 
\bibitem{Karg} M.~I.~Kargapolov, Yu.~I.~Merzlyakov, ``Fundamentals of the Theory of Groups", Springer-Verlag, New York, 1979.
\bibitem{frustrated}I.~P.~Ivanov, V.~Keus,
  Phys.\ Lett.\  {\bf B695}, 459-462 (2011).
\bibitem{Rob}D.~J.~S.~Robinson, ``A Course in the Theory of Groups'', Springer,
1996, Second Edition.
\end{thebibliography}
\end{document}